\documentclass[12pt]{article}
\usepackage[utf8]{inputenc}
\usepackage{fullpage} 
\usepackage[ruled,vlined]{algorithm2e}
\usepackage{amssymb,amsmath,amsthm,mathrsfs}
\usepackage{thmtools,nameref,hyperref,cite,enumitem,cleveref,graphicx,url}
\usepackage{stmaryrd}
\usepackage{kotex} 
\usepackage{natbib}
\usepackage{graphicx}
\usepackage{color}
\usepackage[dvipsnames]{xcolor}
\usepackage{float}
\usepackage[blocks]{authblk}
\usepackage{wrapfig}
\usepackage{verbatim}

\newcommand{\argmin}{\operatornamewithlimits{argmin}}
\newcommand{\BEAS}{\begin{eqnarray*}}
\newcommand{\EEAS}{\end{eqnarray*}}
\newcommand{\BEA}{\begin{eqnarray}}
\newcommand{\EEA}{\end{eqnarray}}
\newcommand{\BIT}{\begin{itemize}}
\newcommand{\EIT}{\end{itemize}}
\newcommand{\BNUM}{\begin{enumerate}}
\newcommand{\ENUM}{\end{enumerate}}

\newcommand{\diag}{\mathrm{diag}}

\newcommand{\indep}{\perp \!\!\!\! \perp}

\newtheorem{theorem}{Theorem}

\newtheorem{assumption}{Assumption}

\newtheorem{lemma}{Lemma}
\newtheorem{fact}{Fact}
\newtheorem*{remark}{Remark}

\newcommand{\bA}{\boldsymbol{A}}
\newcommand{\bB}{\boldsymbol{B}}

\newcommand{\bM}{\boldsymbol{M}}

\newcommand{\bS}{\boldsymbol{S}}

\newcommand{\bU}{\boldsymbol{U}}
\newcommand{\bV}{\boldsymbol{V}}
\newcommand{\bW}{\boldsymbol{W}}
\newcommand{\bX}{\boldsymbol{X}}

\newcommand{\bZ}{\boldsymbol{Z}}

\newcommand{\bbX}{\mathbb{X}}

\newcommand{\bzero}{\boldsymbol{0}}

\newcommand{\bGamma}{\boldsymbol{\Gamma}}
\newcommand{\bSigma}{\boldsymbol{\Sigma}}
\newcommand{\bPi}{\boldsymbol{\Pi}}

\newcommand{\bOmega}{\boldsymbol{\Omega}}

\colorlet{JL}{NavyBlue}
\colorlet{SP}{purple}
\colorlet{MS}{ForestGreen}

\title{Graph Estimation Based on Neighborhood Selection for Matrix-variate Data}

\author[1]{Minsub Shin}

\author[1]{Johan Lim}

\author[2,3]{Seongoh Park\footnote{To whom all correspondence should be addressed. Email: \texttt{spark6@sungshin.ac.kr}}}

\affil[1]{Department of Statistics, Seoul National University, Seoul, Korea} 
\affil[2]{School of Mathematics, Statistics and Data Science, Sungshin Women's University, Seoul, Korea} 
\affil[3]{Data Science Center, Sungshin Women's University, Seoul, Korea}

\date{} 

\date{\today}

\begin{document}
\maketitle

\baselineskip 24pt

\maketitle

\begin{abstract} 
	\noindent  
	Undirected graphical models are powerful tools for uncovering complex relationships among high-dimensional variables. This paper aims to fully recover the structure of an undirected graphical model when the data naturally take matrix form, such as temporal multivariate data.
	As conventional vector-variate analyses have clear limitations in handling such matrix-structured data, several approaches have been proposed, mostly relying on the likelihood of the Gaussian distribution with a separable covariance structure. Although some of these methods provide theoretical guarantees against false inclusions (i.e. all identified edges exist in the true graph), they may suffer from crucial limitations: (1) failure to detect important true edges, or (2) dependency on conditions for the estimators that have not been verified.
	We propose a novel regression-based method for estimating matrix graphical models, based on the relationship between partial correlations and regression coefficients. Adopting the primal-dual witness technique from the regression framework, we derive a non-asymptotic inequality for exact recovery of an edge set. Under suitable regularity conditions, our method consistently identifies the true edge set with high probability. 
	Through simulation studies, we compare the support recovery performance of the proposed method against existing alternatives. We also apply our method to an electroencephalography (EEG) dataset to estimate both the spatial brain network among 64 electrodes and the temporal network across 256 time points.

    \vskip0.5cm 
	\noindent {\bf Keywords:} 
	Gaussian graphical model; Kronecker product; undirected graph; support recovery; nodewise regression; matrix data; non-asymptotic inequality.
\end{abstract} 
\baselineskip 18pt

\section{Introduction}\label{sec:intro}

High-dimensional data are now commonplace in statistical analyses, where the sample size is much smaller than the data dimension. In particular, researchers are increasingly interested in matrix-variate data, where the observation from each subject is naturally represented as a matrix.
For example, time-course gene expression data involve repeated measurements of approximately $10,000$ genes across multiple time points \citep{coffey2014clustering, phat2023alterations}, whereas other studies \citep{zahn2007agemap} measure gene expression across different tissue types. In electroencephalography (EEG) data \citep{begleiter1995eeg}, signals are recorded from 64 electrodes at 256 Hz for each second, resulting in $64 \times 256$ measurements per run. In most such cases, the number of subjects is fewer than one hundred—far smaller than the data dimension, denoted by $d = p \times q$ where $p$ and $q$ represent the row and column dimensions of the data matrix, respectively.
In undirected graphical models, which are closely related to inverse covariance matrices, the number of parameters to estimate increases on the order of $d^2$. A straightforward approach to estimate the inverse covariance matrix is to vectorize it and apply standard vector-based methods, such as the graphical lasso (\citet{Friedman:2007}) or neighborhood selection (\citet{m&b2006}). However, this strategy leads to significant challenges both in computational and theoretical aspects, since the order of $d^2$ is often computationally intractable and the vectorization method fails to capture the intrinsic structure of matrix-variate data. In contrast, directly using the matrix structure helps overcome these challenges and improves estimation efficiency.

While matrix-variate data analysis has recently gained attention in statistics and computer science due to its advantages, its origin dates back to the matrix normal distribution \citep{Dawid:1981}. It models a random matrix $\bX \in \mathbb{R}^{p \times q}$ as $\bX \sim N_{p,q}(\bM, \bU, \bV)$, where $\bM \in \mathbb{R}^{p \times q}$ is a mean matrix and $\bU\in \mathbb{R}^{p \times p},\bV\in \mathbb{R}^{q \times q}$ are positive definite matrices. Equivalently, 
$$
\text{vec}(\bX) \sim N_{pq}(\text{vec}(\bM), \bV \otimes \bU),
$$
where the vectorization of a matrix $\text{vec}(\cdot)$ stacks columns into a vector and $\otimes$ is the Kronecker product between two matrices. Since the covariance matrix of $\text{vec}(\bX)$ is factorized into row-wise and column-wise components, i.e. $\bU$ and $\bV$, parameter estimation and interpretation become simpler and more intuitive.

The matrix-variate graphical model aims to estimate the precision matrices $\bOmega_U=\bU^{-1}$ and $\bOmega_V=\bV^{-1}$, since the 0-1 patterns of these matrices imply the conditional independence structure among variables. Consequently, $\ell_1$-regularized likelihood methods have been proposed in literature with slight modifications in each. Given centered matrix-variate normal data $\{\bX^{(i)}\}_{i=1}^n$,
\citet{Leng:2012}, \citet{Yin:2012}, and \citet{Tsiligkaridis:2013} proposed a method that alternately minimizes the following two penalized likelihood functions with a certain regularizer $p(\cdot)$:
\begin{equation}\label{eq:flip-flop}
	\begin{array}{c}
	\min \limits_{\bOmega_U \succ 0} \{ -\log|\bOmega_U| + \text{tr}( \tilde{\bU} \bOmega_U) + \lambda_U \cdot p(\bOmega_U) \},\\
	\min\limits_{\bOmega_V \succ 0} \{ -\log|\bOmega_V| + \text{tr}( \tilde{\bV} \bOmega_V) + \lambda_V \cdot p(\bOmega_V) \},
	\end{array}
\end{equation}
where $\tilde{\bU}=\sum_{i=1}^n \bX^{(i)} \bOmega_V (\bX^{(i)})^\top / (nq)$ and $\tilde{\bV}=\sum_{i=1}^n (\bX^{(i)})^\top \bOmega_U \bX^{(i)} / (np)$. This approach is essentially to minimize a negative log-likelihood function with appropriate regularization. 	 
They proved the asymptotic convergence rate of local minimizers in the Frobenius norm under several regularity conditions such as bounded eigenvalues, appropriate scaling of the tuning parameter, and the specific rate of sample size and dimension. 

Compared to the aforementioned approaches, which require iterative solutions to the graphical lasso problem (\citet{Friedman:2007}), \citet{Zhou:2014} proposed a non-iterative algorithm called ``GEMINI''. Instead of using $\tilde{\bU}$ or $\tilde{\bV}$ in \eqref{eq:flip-flop}, GEMINI combines the sample covariance matrices computed from row-binding or column-binding of the data with the graphical lasso:
$\tilde{\bS}_U=\sum_{i=1}^n \bX^{(i)} (\bX^{(i)})^\top / (nq)$ or $\tilde{\bS}_V=\sum_{i=1}^n (\bX^{(i)})^\top \bX^{(i)} / (np)$ (in fact, their correlation versions).
The author established the convergence rates of the GEMINI estimator in both the spectral norm and the Frobenius norm. Similar results were obtained by \citet{Wu:2023} under a weak sparsity assumption, using the Sparse Column-wise Inverse Operator (SCIO) \citep{Liu:2015} instead of the graphical lasso.

The primary goal of the graphical model is to uncover the edges in the graph, which is to identify non-zero elements of the precision matrix in the case of normal data. This support recovery problem has been less investigated in the context of matrix graphical models, and only two of the aforementioned studies provide theoretical guarantees for it. More precisely, \citet{Leng:2012} and \citet{Yin:2012} showed that their estimators achieve no false inclusions—that is, all estimated edges belong to the true edge set. However, these results are partial because some true edges may be dropped in the estimated graph, implying that significant connections can be undetected in applications. Furthermore, their theoretical guarantees rely on the assumption that the inverse covariance estimators satisfy a certain rate of convergence in the spectral norm. However, since they did not analyze this spectral norm convergence rate, the assumption remains unverified, casting doubt on the practical applicability of their work.

To fill this gap, we propose an extension of the neighborhood selection \citep{m&b2006} to matrix-variate data and develop a statistical theory for exact support recovery (Theorem \ref{mainthm}). Specifically, we regress each row variable on the remaining rows with $\ell_1$-regularization and read the zeros of coefficients to recover the zero pattern of $\bOmega_U$, and repeat the procedure along the columns to recover $\bOmega_V$.
This regression-based framework is attractive because it allows us to adopt established tools from high-dimensional regression—such as primal-dual witness constructions—to prove support recovery under the incoherence condition (Assumption \ref{assum:alpha}). Moreover, we derive a novel non-asymptotic bound in the spectral norm for the sample covariance matrix formed by stacking columns of the observations (Lemma~\ref{lem:Gamma_SS_dev}), which is completely new and stands on its own interest.

The rest of the paper is organized as follows.
In Section 2, we describe the neighborhood selection procedure for matrix-variate data, including details on estimation of support and selection of hyperparameters. We present theoretical properties and required assumptions in Section 3. Section 4 presents the empirical performance of proposed method through simulation studies, and compares it with other comparative methods. In Section 5, we apply the method to EEG data to demonstrate its applicability to real-world scenarios.
Finally, Section 6 concludes the paper with a discussion.

We conclude this section by briefly introducing the notations used throughout this paper. 
Let $[p]=\{1,2,\ldots,p\}$ for a positive integer $p$. Denote by $|S|$ the cardinality of a set $S$. For $k \in [p]$, define $(k)= \{1, \ldots, k-1, k+1, \ldots, p\}$.
For a vector $x \in \mathbb{R}^p$ and $k \in [p]$, denote a leave-one-out vector by $x_{(k)}=(x_1, \ldots, x_{k-1},x_{k+1} \ldots, x_p)^\top \in \mathbb{R}^{p-1}$. For $p \times q$ matrix $\bX=(X_1, \ldots, X_p)^\top$, $\bX_{(k)}$ is a $(p-1) \times q$ matrix where $X_k$ is eliminated. Moreover, for $S_1 \subset [p]$ and $S_2 \subset [q]$, let $\bX_{S_1 S_2}$ be the sub-matrix consisting of rows and columns indexed by $S_1$ and $S_2$, respectively. For a singleton set $S=\{s\}$, we use $s$ instead of $S$ to denote the sub-matrix. For example, $X_{s}$ denotes the $s$-th row of $\bX$.

We use $\Vert \cdot \Vert _p$ for both vector and matrix, as long as they do not cause confusion. For a vector $x \in \mathbb{R}^n$, let $ \Vert x \Vert _p=(\sum_{i=1}^n |x_i|^p)^{1/p}$ for $1 \leq p 
<\infty$ and $ \Vert x \Vert _{\infty}=\max_{1 \leq i \leq n}|x_i|$.
For a matrix $\bA \in \mathbb{R}^{p \times q}$, denote by $\Vert \bA \Vert _2=\sup_{x \neq 0} \Vert \bA x \Vert _2 / \Vert x \Vert _2$ the spectral norm of $\bA$, and by $\Vert \bA \Vert _F=\sqrt{\sum_{i=1}^p \sum_{j=1}^q a_{ij}^2}$ the Frobenius norm of $\bA$. Note that $ \Vert \bA \Vert _2$ is equal to the largest singular value of $\bA$, and $ \Vert \bA \Vert _2 \leq  \Vert \bA \Vert _F \leq \sqrt{\text{rank}(\bA)} \Vert \bA \Vert_2$ holds in general.
For $p \times p$ matrix $\bB$, denote the trace of $\bB$ by $\text{tr}(\bB)$, and the maximum and minimum eigenvalues by $\lambda_{\max}(\bB)$ and $\lambda_{\min}(\bB)$, respectively. While $\diag(\bB)$ denotes the vector $(b_{11}, b_{22}, \ldots, b_{pp})^\top$, $\diag_n(\bB)$ denotes the $np \times np$ block matrix with all diagonal blocks equal to $\bB$. Throughout this paper, we use bold symbols (e.g. $\bX$) to denote matrices and regular symbols (e.g. $X$) to represent vectors or scalars.

\section{Method}\label{sec:method}

\subsection{Matrix-variate data and graph structure}

We consider $p \times q$ matrix normal data with zero mean and a separable covariance structure:
\begin{equation}\label{eq:matnormal}
\bX \sim MN_{p, q}(\bzero_{p\times q}, \bU_{p \times p}, \bV_{q \times q}),
\end{equation}
which means $\text{vec}(\bX)$ is normally distributed with zero mean and covariance matrix $\bV \otimes \bU$, where $\bU \in \mathbb{R}^{p\times p}$ and $\bV \in \mathbb{R}^{q\times q}$ denote row-wise and column-wise covariance matrices, respectively. 
By the Markov property of the multivariate normal distribution, the graph structure can be inferred from the inverse covariance matrix of $\bX$, which is $\bOmega = \bV^{-1} \otimes \bU^{-1}$. More specifically, $(\bU^{-1})_{ab}=0$ implies the conditional independence of row vectors $X_{a}$ and $X_{b}$, given all other rows. In the graph $G=(N,E)$, where $N=[p] \times [q]$ and $E=\{(i,j) \in N \times N: i\neq j \}$, this corresponds to the absence of edges between the node sets $\{(a,v): v \in [q] \}$ and $\{(b,v): v \in [q] \}$. This paper aims to estimate the sparsity patterns of the precision matrices $\bU^{-1}$ and $\bV^{-1}$, thereby recovering the graph structure of matrix-variate data.

\subsection{Neighborhood selection for matrix-variate data}\label{subsection:method}

We extend the neighborhood selection framework proposed by \citet{m&b2006} to matrix-variate data. For this purpose, we consider the regression of $X_{a}$ on $\bX_{(a)}$ and define the coefficient vector that minimizes the mean squared error:
$$
\theta^a=\argmin_{\theta \in \mathbb{R}^{p-1}}\mathbb{E} \Vert X_{a}-\bX_{(a)}^\top \theta \Vert _2^2=\mathbb{E}[\bX_{(a)}\bX_{(a)}^\top]^{-1}\mathbb{E}[\bX_{(a)} X_{a}].
$$
From $\mathbb{E}[\bX\bX^\top]=\text{tr}(\bV)\bU$, we can see 
\begin{equation}\label{eq:relation_coef_precision}
\theta^a= \big\{ \text{tr}(\bV)\bU_{(a),(a)} \big\}^{-1} \big\{\text{tr}(\bV)\bU_{(a),a} \big\} = -\, \dfrac{(\bU^{-1})_{(a),a}}{(\bU^{-1})_{aa}},
\end{equation}
where the last equality uses the inversion of block matrices.
Thus, the regression coefficients $\{\theta^a\}_{a\in [p]}$ directly encode the sparsity pattern of $\bU^{-1}$. We adopt the $\ell_1$-regularization to induce sparsity in the coefficients. Let $\bX^{(1)}, \ldots, \bX^{(n)}$ be $n$ i.i.d. samples from \eqref{eq:matnormal}. Then for each $a \in [p]$, we solve
\begin{equation}\label{eq:lassomodel}
\hat{\theta}^a = \argmin_{\theta \in \mathbb{R}^{p},\, \theta_a = 0} \left\{ \frac{1}{2n}\sum_{i=1}^{n} \Vert X_{a}^{(i)} - (\bX^{(i)})^\top \theta \Vert _2^2+\lambda \Vert \theta \Vert _1 \right\},
\end{equation}
where $\lambda > 0$ is a tuning parameter.

We estimate the neighborhood of each variable $a$ as $\widehat{\mathcal{N}}(a)=\{b \in [p] : \hat{\theta}^{a}_b \neq 0 \}$. Due to the asymmetry of the neighborhood selection procedure, the estimated edge set can be constructed using either the AND or OR rule, as in the earlier work of \citet{m&b2006}: $\hat{E}_{\text{AND}}:=\{(a,b) \in [p] \times [p]: a \in \widehat{\mathcal{N}}(b) \text{ and } b \in \widehat{\mathcal{N}}(a)\}$ or $\hat{E}_{\text{OR}}:=\{(a,b) \in [p] \times [p]: a \in \widehat{\mathcal{N}}(b) \text{ or } b \in \widehat{\mathcal{N}}(a)\}$.

In Section \ref{sec:theory}, we show that the row-wise edge set $E_U=\{(a,b) \in [p] \times [p] : (\bU^{-1})_{ab} \neq 0 \}$ can be consistently recovered using either the AND or OR rule. The column-wise estimators can be obtained by applying the same procedure to the transposed data.
We present several remarks on the proposed method.
\begin{remark}
\text{ }
	\begin{enumerate}
        \item Note that the matrices $\bU$ and $\bV$ in \eqref{eq:matnormal} are not jointly identifiable, since any scalar factor may be transferred between them. Nonetheless, their graph structures are preserved up to a scalar multiple. Therefore, the neighborhood selection procedure is not affected by this identifiability issue. In particular, the edge estimate $\hat{E}$ derived from \eqref{eq:lassomodel} is well-defined for any admissible $\bU$ and satisfies the same theoretical properties.

		\item The row-wise and column-wise estimations can be performed independently. In contrast, likelihood-based approaches (e.g. \citet{Leng:2012,Yin:2012,Tsiligkaridis:2012}) involve alternating steps between updating row-wise and column-wise estimates, which can be computationally demanding.

		\item 
		The row-wise neighborhood selection in \eqref{eq:lassomodel} is equivalent to applying the standard neighborhood selection to the $nq$ vectors (i.e. column vectors of $\{X^{(i)}\}_{i \in [n]}$). This means that the penalized least-squares treats these vectors as if they were independent observations. Therefore, the working sample size is $nq$, not just $n$, which is referred to as the ``effective sample'' in earlier works (e.g. \citet{Leng:2012,Zhou:2014,Ning:2013,Chen:2019}). Our theory also confirms this interpretation; see the second remark following Theorem \ref{mainthm}.

		\item
		Existing theories, such as Theorem 1 in \citet{m&b2006}, are not directly applicable here due to the correlation among the $nq$ vectors. To address this challenge, we use the Hanson-Wright inequality (see \eqref{eq:HWineq} in Appendix) to derive a new concentration inequality for the sample covariance matrix formed by the $nq$ dependent vectors (see Lemma~\ref{lem:Gamma_SS_dev}).
		
		\item 
		In practice, we recommend a standardization process before applying the proposed method, i.e., transforming all $pq$ variables to have empirical mean zero and unit standard deviation.
		
	\end{enumerate}
\end{remark}

\subsection{Tuning parameter selection}\label{sec:tuning}

As our method relies on regularization, careful selection of the penalty parameter is crucial.
Given the nature of neighborhood selection, optimal tuning can be approached as in the standard regression framework. For example, consider the penalized regression in \eqref{eq:lassomodel} where node $a$ is regressed on the others. To determine an optimal hyperparameter, we split the dataset into $K$ folds and conduct a cross-validation procedure. Among a candidate set $\Lambda$, we select the optimal value $\hat{\lambda}_a\in \Lambda$ that minimizes the cross-validated prediction error for each node $a \in [p]$. However, this node-wise tuning can be computationally intensive, especially when $p$ is large.

For better efficiency, we instead use a single $\lambda_{\text{row}}$ for all row-wise regressions and another $\lambda_{\text{col}}$ for all column-wise regressions. We empirically compare two tuning procedures; (1) an ``individual'' procedure, which searches a separate hyperparameter for each node, and (2) a ``global'' procedure, which searches only two parameters $\lambda_{\text{row}}$ and $\lambda_{\text{col}}$. Our numerical study in Section \ref{subsec:tuning} demonstrates that the two procedures achieve comparable performance in graph recovery, thereby suggesting that the global procedure can be a preferable option.

\section{Theoretical results}\label{sec:theory}

In this section, we present our main result—consistency of the graph structure recovered by the proposed method—and the required assumptions.

First, we introduce the notion of $\alpha$-incoherence, also known as the irrepresentability condition. 
Let $\bGamma \in \mathbb{R}^{r \times r}$ be a positive definite matrix and $S$ be a subset of $[r]$. We say that $\bGamma$ is $\alpha$-incoherent with respect to $S$ if
$\max_{b \in S^c}  \Vert \bGamma_{bS}(\bGamma_{SS})^{-1} \Vert _1 \leq 1-\alpha$,
for given $\alpha \in (0,1]$. In our context, we require this condition to hold for $\bU_{(a),(a)}$ with respect to the neighborhood set $\mathcal{N}(a):=\{b \in [p]\setminus \{a\}:(\bU^{-1})_{ab} \neq 0 \}$, for each $a \in [p]$.

\begin{assumption}[$\alpha$-incoherence] \label{assum:alpha}
There exists a constant $\alpha >0$, independent of $p$ and $q$, such that $\bU_{(a),(a)}$ is $\alpha$-incoherent with respect to $\mathcal{N}(a)$ for all $ a \in [p]$.
\end{assumption}
\noindent
This assumption implies that the dependence between variables in $S$ and those in $S^c$ is sufficiently weak. 
We further assume sparsity in the graph structure. In particular, the degree of each node in the graph is bounded by a constant, even as $p$ or $q$ may grow.

\begin{assumption}[Maximum degree] \label{assum:sparse}
    There exists a constant $d_{\max} \in \mathbb{N}$, independent of $p$ and $q$, such that $|\mathcal{N}(a)| \leq d_{\max}$ for all $a \in [p]$.
\end{assumption}
\noindent
As in many related works, we assume that the eigenvalues of the covariance matrices are uniformly bounded. 

\begin{assumption}[Bounded eigenvalues] \label{assum:eigen}
There exist constants $\lambda_{\min}^U, \lambda_{\min}^V, \lambda_{\max}^U, \lambda_{\max}^V>0$, independent of $p$ and $q$, such that
$\lambda_{\min}^U \leq \lambda_{\min}(\bU) \leq \lambda_{\max}(\bU) \leq \lambda_{\max}^U$ and $\lambda_{\min}^V \leq \lambda_{\min}(\bV) \leq \lambda_{\max}(\bV) \leq \lambda_{\max}^V$.
\end{assumption}

We also define constants $u_{\max}:=\max_{1 \leq j \leq p}u_{jj}$ and $v_{\text{avg}}:=\text{tr}(\bV)/q$.
Based on these assumptions, we now introduce our main theorem.

\begin{theorem}[Support recovery] \label{mainthm}
Suppose Assumptions \ref{assum:alpha}, \ref{assum:sparse}, and \ref{assum:eigen} hold. Let $\bX^{(1)}, \ldots, \bX^{(n)}$ be independent samples from \eqref{eq:matnormal}. For any constant $\beta >1$, choose 
$$\lambda \geq c \sqrt{\beta \cdot v_{\text{avg}} \cdot \lambda_{\max}^V \frac{\log p}{nq}} \max \left( \frac{u_{\max}}{\alpha}, \sqrt{\frac{\lambda_{\min}^U}{d_{\max}}} \right)$$
for some constant $c>0$. Moreover, assume either conditions (a), (b) and (c1) hold or (a), (b) and (c2) hold, for some constants $c_1, c_2 > 0$.
\begin{enumerate}
    \item[(a)]  $\displaystyle\frac{nq}{\log p} \geq c_1 \beta \cdot\max \left(\frac{d_{\max}}{\alpha^2} \cdot\frac{u_{\max}}{\lambda_{\min}^U} \frac{\lambda_{\max}^V}{v_{\text{avg}}}
    , \left( \frac{ \lambda_{\max}^V}{ v_{\text{avg}}}\right)^2\right)$,
    
    \item[(b)] $\displaystyle n \left( \frac{2d_{\max}}{q} + 1 \right) \geq (\lambda_{\max}^V)^2$,
    
    \item[(c1)] $\displaystyle \frac{n}{\log p} \ge c_2^2 \beta \left(\frac{2d_{\max}}{q} + 1 \right) \cdot\frac{ \lambda_{\max}^U}{\lambda_{\min}^U} \left( \frac{ \lambda_{\max}^V}{ v_{\text{avg}}}\right)^2, \quad
	\frac{n}{q^2} \le  c_2^4 \left( \frac{2d_{\max}}{q}  + 1 \right)^3 \cdot \left( \frac{\lambda_{\max}^U}{\lambda_{\min}^U }\right)^4 \frac{(\lambda_{\max}^V)^2}{v_{\text{avg}}^4}$,
    
    \item [(c2)] $\displaystyle \frac{nq}{(\log p)^{3/2}} \ge c_2\beta^{3/2} \cdot \frac{\lambda_{\max}^U }{\lambda_{\min}^U } (\lambda_{\max}^V)^2, \quad
	\frac{n}{q^2} >  c_2^4 \left( \frac{2d_{\max}}{q}  + 1 \right)^3 \cdot\left( \frac{\lambda_{\max}^U}{\lambda_{\min}^U }\right)^4\frac{(\lambda_{\max}^V)^2}{v_{\text{avg}}^4}$.
\end{enumerate}
\noindent
Based on the matrix neighborhood selection \eqref{eq:lassomodel}, estimate the edge set $\hat{E}$ by either
$\hat{E}_{\text{AND}}$ or $\hat{E}_{\text{OR}}$ as in Section \ref{subsection:method}.
Then, with probability at least $1- \max(3d_{\max}+6,\ 5d_{\max})/p^{\beta-1}$, the estimated edge set $\hat{E}$ successfully recovers the true set $E_U=\{(a,b) \in [p] \times [p] : (\bU^{-1})_{ab} \neq 0 \}$ according to the following criteria:
\begin{itemize}
	\item[(i)] No false inclusions: $\hat{E} \subset E_U$,
	\item[(ii)] No false exclusions:
	$\hat{E} \supset E_U$ if $\min_{(a,b) \in E_U}|(\bU^{-1})_{ab}| \geq 3\lambda\sqrt{d_{\max}}/(\lambda_{\min}^Uv_{\text{avg}})$.
\end{itemize}
\end{theorem}

The core of the proof relies on the primal-dual witness argument, a standard technique for establishing variable selection consistency in linear regression. To apply the technique, one must handle a random design matrix, which requires analyzing the row-wise sample covariance matrix defined by $\sum_{i=1}^n \bX^{(i)}(\bX^{(i)})^{\top}/nq$. Lemma \ref{lem:Gamma_SS_dev} provides a new non‐asymptotic spectral norm bound for this matrix, which may be of independent interest. Another challenge in the theoretical analysis is the correlation among the columns of $\bX^{(i)}$, which is also considered in the derivation. The complete proof is provided in \ref{sec:pf_mainthm}. Note that results for column-wise estimation can be obtained similarly by adjusting the assumptions accordingly.

Theorem \ref{mainthm} implies several interesting findings. 
\begin{enumerate}
    \item 
    The ``no false exclusions" result is new in the literature on matrix graphical models. It suggests that the proposed method can detect all signals whose magnitudes exceed a certain threshold. This condition is adapted from the ``beta-min'' condition frequently used in the regression literature, and is also required in the neighborhood selection for vector-variate data (see Assumption 5 in \citet{m&b2006}). The magnitude of the threshold decreases —meaning more true edges are detectable— as (1) $p$ decreases and $n,q$ increase, given the choice of $\lambda$ on the order of $\sqrt{\log p / nq}$, and (2) the maximum degree decreases.

    \item 
    We use $a \gtrsim b$ to indicate that $a$ is greater than $b$ up to a constant factor, independent of $n$, $p$, and $q$. The dimension regime (c1) roughly corresponds to $A \cap B$, while (c2) corresponds to $A^c \cap C$, where
    $$
    A=\{q^2 \gtrsim n\}, \quad B=\{n \gtrsim \log p\}, \quad C=\{nq \gtrsim (\log p)^{3/2}\}.
    $$
    Since $A \cap B \subset C$ and $A^c \cap C \subset B$, the union of $(c1)$ and $(c2)$ simplifies to $B \cap C$, or equivalently,
    $$
    n \gtrsim \log p, \quad nq \gtrsim (\log p)^{3/2}.
    $$
    As a consequence, when $q$ is fixed, at least $O\big((\log p)^{3/2} \big)$ samples are required. If $q$ is of order $(\log p)^{1/2}$, fewer samples of order $O(\log p)$ suffice. This highlights that the effective sample size for estimating the row-wise graph structure is indeed $ nq $.
    However, this regime does not cover the case where $n$ is fixed and $q$ diverges.

    \item
    Regarding the sample size condition, a close look at related results in the literature concludes that no false inclusion typically requires $p \log p = o(n)$ (\citet{Yin:2012,Leng:2012}) or $p \log p = o(nq)$ (\citet{Tsiligkaridis:2012,Tsiligkaridis:2013,Ning:2013}), both of which are stronger than ours, which requires only $\log p =o(n)$ and $\log p =o(nq)$.

    \item
    At first glance, the condition on $\lambda$ in Theorem \ref{mainthm} may seem counter-intuitive: as $\lambda_{\min}^U$ moves away from zero, the lower bound on $\lambda$ increases. However, statement (ii) shows that using this optimal admissible $\lambda$ yields a detection threshold proportional to $(\lambda_{\min}^U)^{-1/2}$. 
    Hence, as $\lambda_{\min}^U$ increases, the detection threshold decreases.
    Moreover, the choice of $\lambda$ in the order of $ \lambda \asymp \lambda_{\min}^U $ concludes that the tail probability increases as $ \lambda_{\min}^U $ approaches zero, which is consistent with our intuition (see \eqref{lambda_tailprob} in the proof).

\end{enumerate}

\section{Simulation study}\label{sec:simulation}
\subsection{Setup}
In this section, we report our numerical simulation results to evaluate the performance of proposed graph estimation method for matrix-variate data. We compare our matrix neighborhood selection method (hereafter, matrixNS) with two existing approaches: the GEMINI algorithm of \citet{Zhou:2014} and the penalized‐likelihood flip‐flop algorithm of \citet{Leng:2012}. In each simulation, we generate inverse covariance matrices $\bU^{-1}$ and $\bV^{-1}$ according to the following graph structures, described in \citet{Chen:2019}.

\begin{itemize}
    \item Hub graph (``hub"): For fixed $p$, let $$\omega_{ij}= \begin{cases}
        1 \quad \text{if}\ \ i=j\\
        \rho_{\text{hub}}\quad \text{if}\ i=10(k-1)+1,\ 10(k-1)+2 \leq j \leq 10k\\
        0 \quad \text{otherwise}
    \end{cases}$$
    for $1 \leq k \leq p/10$.
    \item Band graph (``band"): Set $$\omega_{ij}= \begin{cases}
        1 \quad \text{if}\ \ i=j\\
        \rho_{\text{band}}\quad \text{if}\ |i-j|=1\\
        \rho_{\text{band}}/2 \quad \text{if}\ |i-j|=2\\
        0 \quad \text{otherwise}.
    \end{cases}$$
    \item Erd\"{o}s-R\'{e}nyi random graph (``random"): Each off-diagonal entry $\omega_{ij}$ has value $\rho_{\text{rand}}$ with probability $\min(0.05, 5/p)$ independently, and otherwise zero. All diagonal terms are set to $1$.
\end{itemize}
We set the signal size parameters $\rho_{\text{hub}}, \rho_{\text{band}}$, and $\rho_{\text{rand}}$ to 0.4, 0.6, and 0.4 respectively. Different signal sizes are also explored, and the results are presented in Appendix.
For the hub and random $\bOmega$, we modify each matrix by $\Tilde{\bOmega} = \bOmega+ (0.05+|\lambda_{\min}(\bOmega)|)I_p$ to ensure the positive definiteness. 
Moreover, we multiply each diagonal entry $\omega_{ii}$ by an independent random number from the uniform distribution $\mathrm{U}[1,5]$. This induces heterogeneous variances across variables, as commonly encountered in real-world applications. The results for the homogeneous cases are given in Section C.1 of Appendix.
Figure \ref{fig:structure} visualizes the graph structures for the three types of precision matrices.

\begin{figure}[H]
    \centering
    \includegraphics[width=1\textwidth]{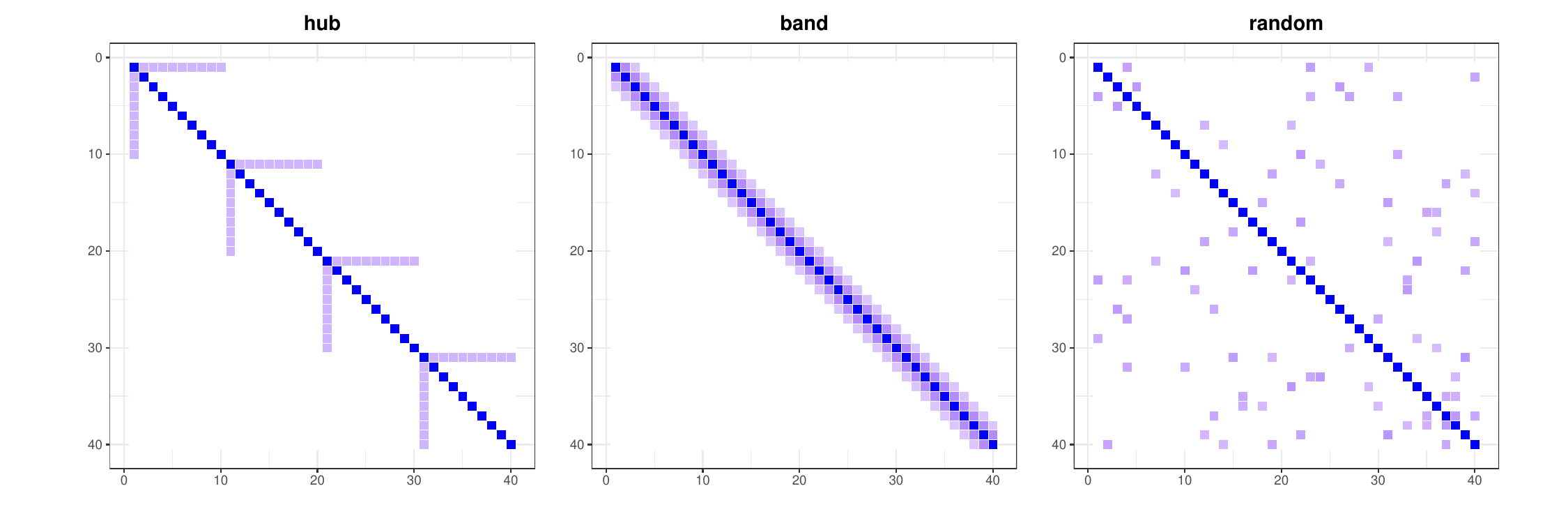}
    \caption{Graph structures for hub (left), band (middle), and Erd\"{o}s-R\'{e}nyi random (right) precision matrices with $p=40$.}
    \label{fig:structure}
\end{figure}

We fix the structure of $\bU^{-1}$ as a band matrix and generate data $\bX$ from $N_{p,q}(\bzero, \bU, \bV)$, while the structure of $\bV^{-1}$ varies among the three types.  We set the sample size to $n=20$ and consider four $(p,q)$ combinations: $(20,20)$, $(20,50)$, $(50,20)$, and $(50,50)$.
For each $(p,q)$ pair and each column graph structure, we generate data $\bX$ and estimate the graph structure using (1) our matrixNS method (``NS"), (2) the GEMINI algorithm \citep{Zhou:2014} (``GEMINI"), (3) and the likelihood flip-flop method \citep{Leng:2012} (``Likelihood"). For the NS method, we standardize the data prior to estimation. In contrast, we use the raw data for the other two, as they directly estimate the precision matrix. We note that standardization for these two methods yields similar results, thus omitted in this paper. Each experimental setting is repeated 100 times.

To evaluate performance, we employ modified ROC curves and AUCs. First, we present ROC curves with the x-axis representing the proportion of zero off-diagonal entries in $\bU^{-1}$ that are nonzero in $\hat{\bU}^{-1}$ (i.e. the false positive rate, FPR) and the y-axis representing the proportion of nonzero off-diagonal entries in $\bU^{-1}$ that are nonzero in $\hat{\bU}^{-1}$ (i.e. the true positive rate, TPR), as the tuning parameter $\lambda$ varies over $2^{-10}, 2^{-9.75}, \cdots, 2^2$. An analogous analysis is performed for $\bV^{-1}$. For the likelihood flip-flop method, we fix the row (or column) tuning parameter at $\lambda=2^{-2}$ when estimating the column (or row) graph, respectively.

In comparing the experimental results, superior performance is indicated by ROC curves closer to the upper-left corner and AUC values approaching 1. However, in typical scenarios, our primary interest lies in cases where the FPR is controlled to be sufficiently low. Hence, we focus on ROC curves for FPR within $[0, 0.15]$, emphasizing performance differences under tight error control (see Section 5.1.2 in \citet{Khare:2015}). The partial AUC is defined by the area under the ROC curve on the range of FPR bewtween 0 and 0.15, after divided by 0.15 to normalizes its scale to be $[0,1]$. We report the partial AUCs and their standard errors on the following graphs for each experiment.

\subsection{Results - comparative methods}

Figure \ref{fig:sim_hetero} presents the ROC curves and the partial AUC values for the three comparative methods. 
Across all settings, our method outperforms the other approaches in estimating edges of the undirected graphs, followed by GEMINI and then Likelihood. However, the results vary depending on the types of graphs. In particular, the difference in AUC is more distinct when the banded structure is estimated; see the top panel that illustrates the performance of row-wise graph estimation and the first row of the bottom panel.
For hub and random structures, there is a modest improvement in the proposed method over GEMINI, but the gap is pronounced compared to Likelihood; see the last two rows of the bottom panel. 

Another notable observation is that the matrix dimensions $p$ and $q$ affect the estimation accuracy, as does the sample size. To be specific, the performance of row-wise graph estimation improves as $q$ grows relative to $p$, and vice versa for column-wise estimation. This supports our theoretical finding regarding the effective sample size, as described in Remark 3 at the end of Section \ref{subsection:method} and in Remark 3 following Theorem \ref{mainthm}.

Finally, we observe that the sample size (i.e. larger $n$) and signal strength (i.e. larger $\rho$ in $\bU^{-1}$ or $\bV^{-1}$) improve estimation accuracy in all scenarios. Due to space limit, these results are provided in Section C.2 and C.3 of Appendix.

%


    


    

\begin{figure}[htbp]
    \centering
    \includegraphics[width=1\textwidth]{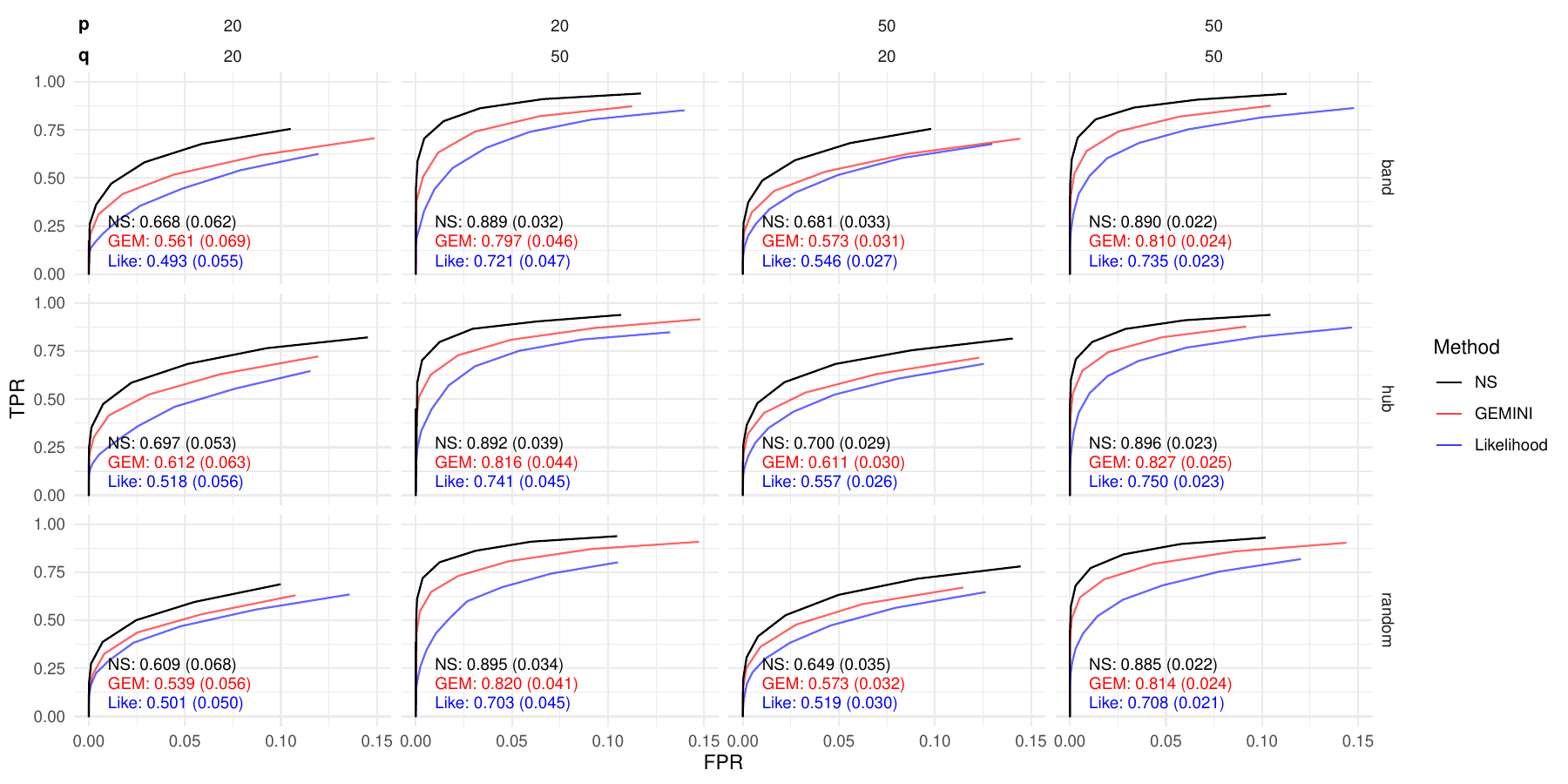}
    \includegraphics[width=1\textwidth]{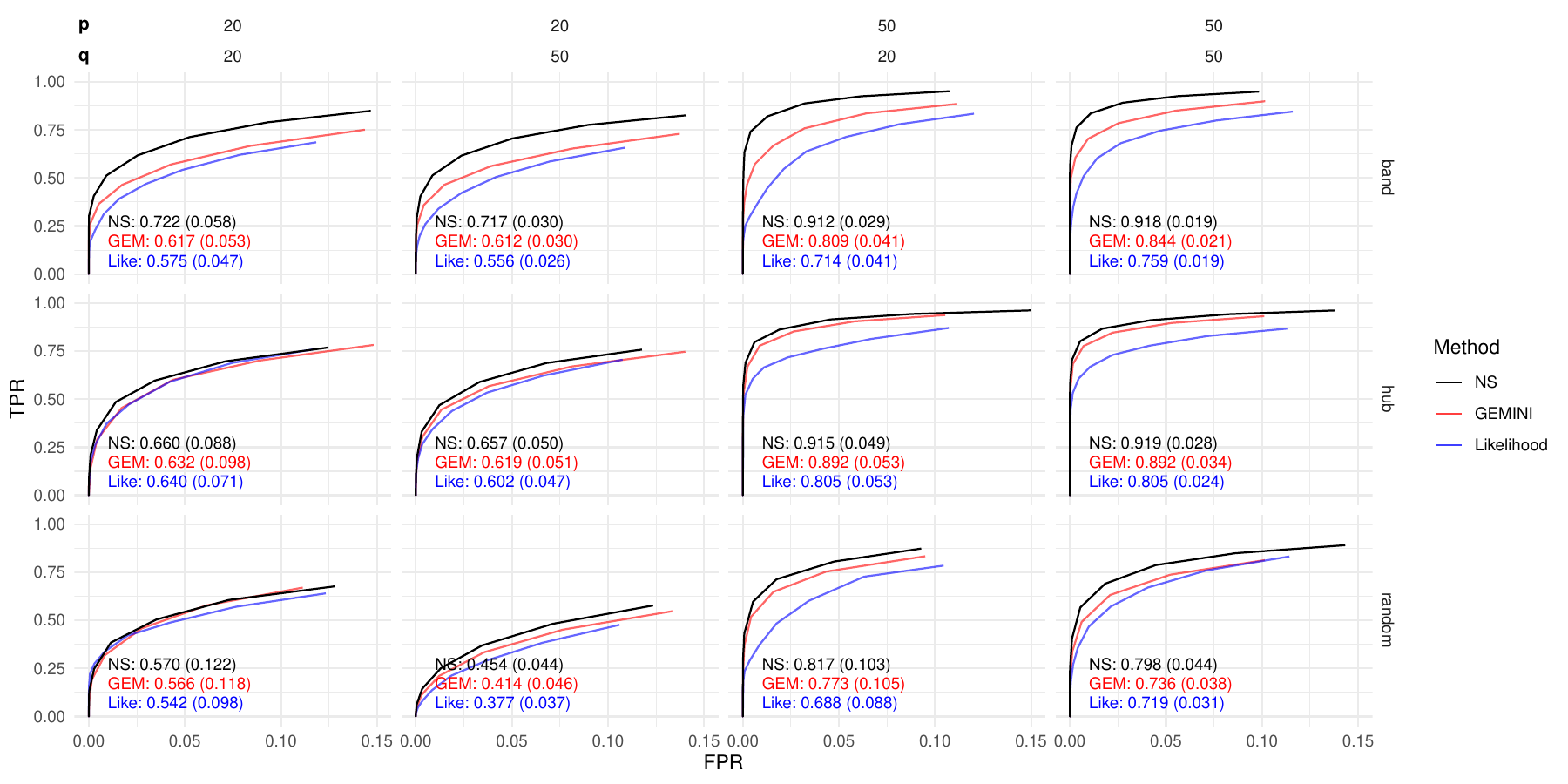}
    \caption{The ROC curves and partial AUC values for the row-wise (top) and column-wise (bottom) graph estimation under heterogeneous variances, comparing the different methods. The horizontal labels represent $(p, q)$ combinations, while the vertical labels indicate true graph structure. Texts within each panel show the partial AUC values and the standard deviation (in parenthesis).}
    \label{fig:sim_hetero}
\end{figure}

\subsection{Results - tuning procedures}\label{subsec:tuning}
We empirically compare two procedures to select the tuning parameters of $\ell_1$-regularization; (1) the ``individual'' procedure that selects a separate hyperparameter for each node-wise regression, and (2) the ``global'' procedure that uses one tuning parameter for all row-wise regressions and another for all column-wise regressions. Note that ``individual'' procedure requires $p+q$ parameters to be tuned, while the ``global'' procedure only has two parameters.

The ``individual'' procedure estimates a row-wise graph using $p$ node-wise regression, each of which is optimally tuned as described in Section \ref{sec:tuning}. Then, we summarize its performance by a single pair of FPR and TPR. On the other hand, we obtain the ROC curve from the ``global'' procedure as in Figure \ref{fig:sim_hetero}. As a result, 
we show the single point from the ``individual'' procedure on the ROC curve for comparison, which is given in Figure \ref{fig:indiv}. 
In all cases, blue dots representing the ``individual'' procedure lie on or below the ROC curve produced by the ``global'' procedure. This result suggests that tuning each regression separately does not necessarily lead to a substantial improvement in performance. Thus, unnecessary computation can be avoided in practice.

\begin{figure}[htbp]
    \centering
    \includegraphics[width=1\textwidth]{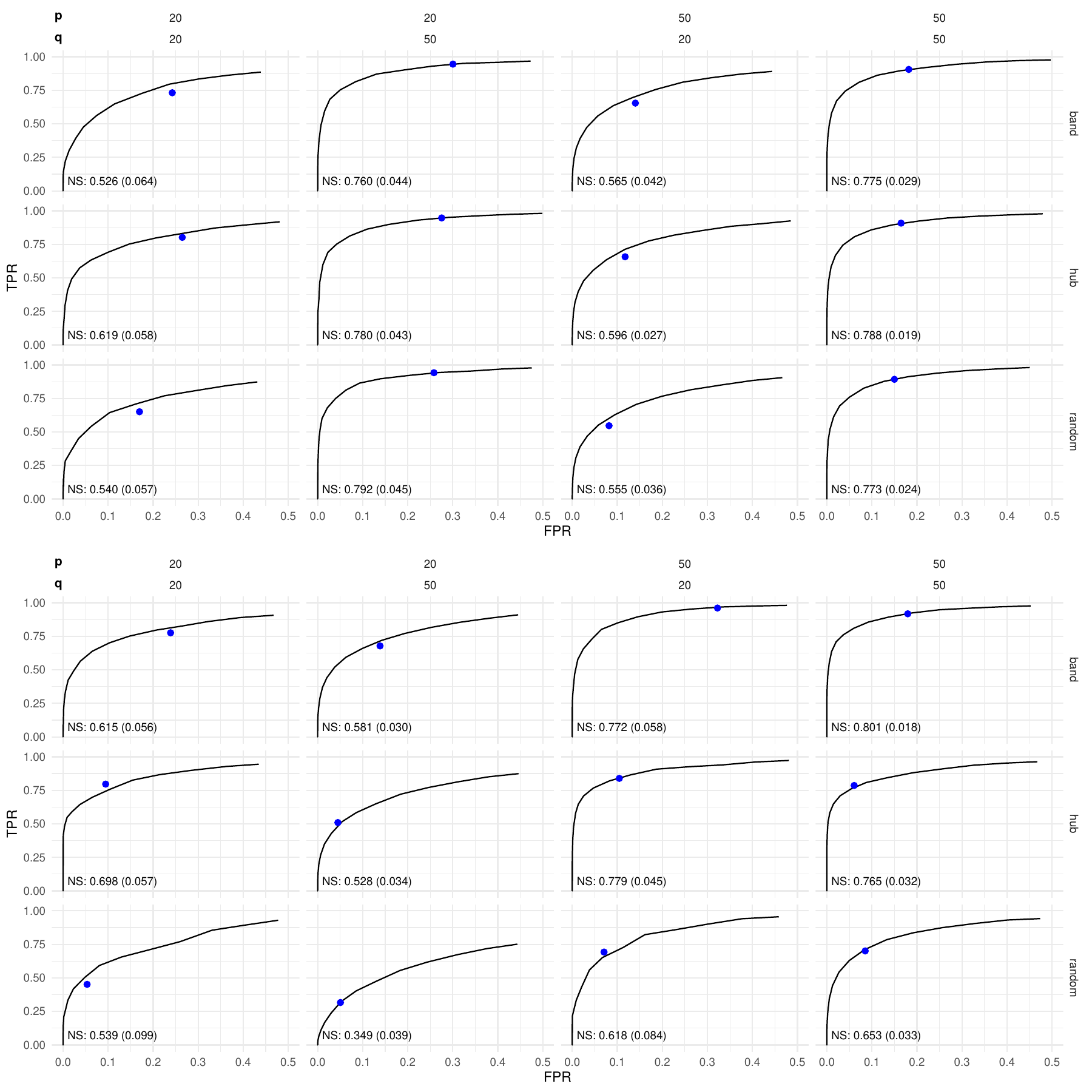}
    \caption{The ROC curves for row-wise (top) and column-wise (bottom) graph estimation. The solid lines represent the ROC curves from the matrixNS method with the ``global'' tuning, while the blue dots indicate the FPR and TPR of the matrixNS method with the ``individual'' tuning.}
    \label{fig:indiv}
\end{figure}

\section{Real data analysis}\label{sec:realdata}

In this section, we analyze an electroencephalography (EEG) dataset from the UCI Machine Learning Repository \citep{begleiter1995eeg} to investigate the underlying structure of brain networks. The data were collected as part of the Collaborative Studies on Genetics of Alcoholism (COGA) project \citep{zhang1995event}, which aimed to explore the relationship between EEG signals and genetic predisposition to alcoholism.

For each subject, the data consist of EEG recordings measured from 64 scalp electrodes, sampled at 256 Hz (i.e., one measurement every 3.9 milliseconds), yielding a matrix observation $p=64$ and $q=256$. Under each stimulus paradigm, recordings were collected across 10 repeated trials. The full dataset contains EEG data for 122 individuals, classified as alcoholic or control group. Here, we demonstrate our method on one control subject under a specific paradigm (ID: \texttt{c\_1\_co2c0000337}), resulting in a setting with $n=10$ trials, $p=64$ electrodes, and $q=256$ time points. Applying matrixNS and GEMINI to this dataset, we estimate the graph structures. For pre-processing, we standardize each of $pq$ variables for matrixNS, while we center them for GEMINI.

Figure \ref{fig:EEGrow} displays the estimated row-wise graphs, which represent the connectivity structure among the $p=64$ electrodes. Node positions follow the Standard Electrode Position Nomenclature \citep{zhang1995event}. We use distinct colors to indicate electrode positions on the left, middle, and right regions of the scalp, following \citet{Zhou:2014}. To ensure a fair comparison between methods, we control the edge density of each estimated graph. Specifically, the connectivity level—defined as the ratio of the number of edges to the total number of possible edges—is approximately fixed at 0.15, 0.10, and 0.05.
\begin{figure}[htbp]
    \centering
    \includegraphics[width=0.8\textwidth]{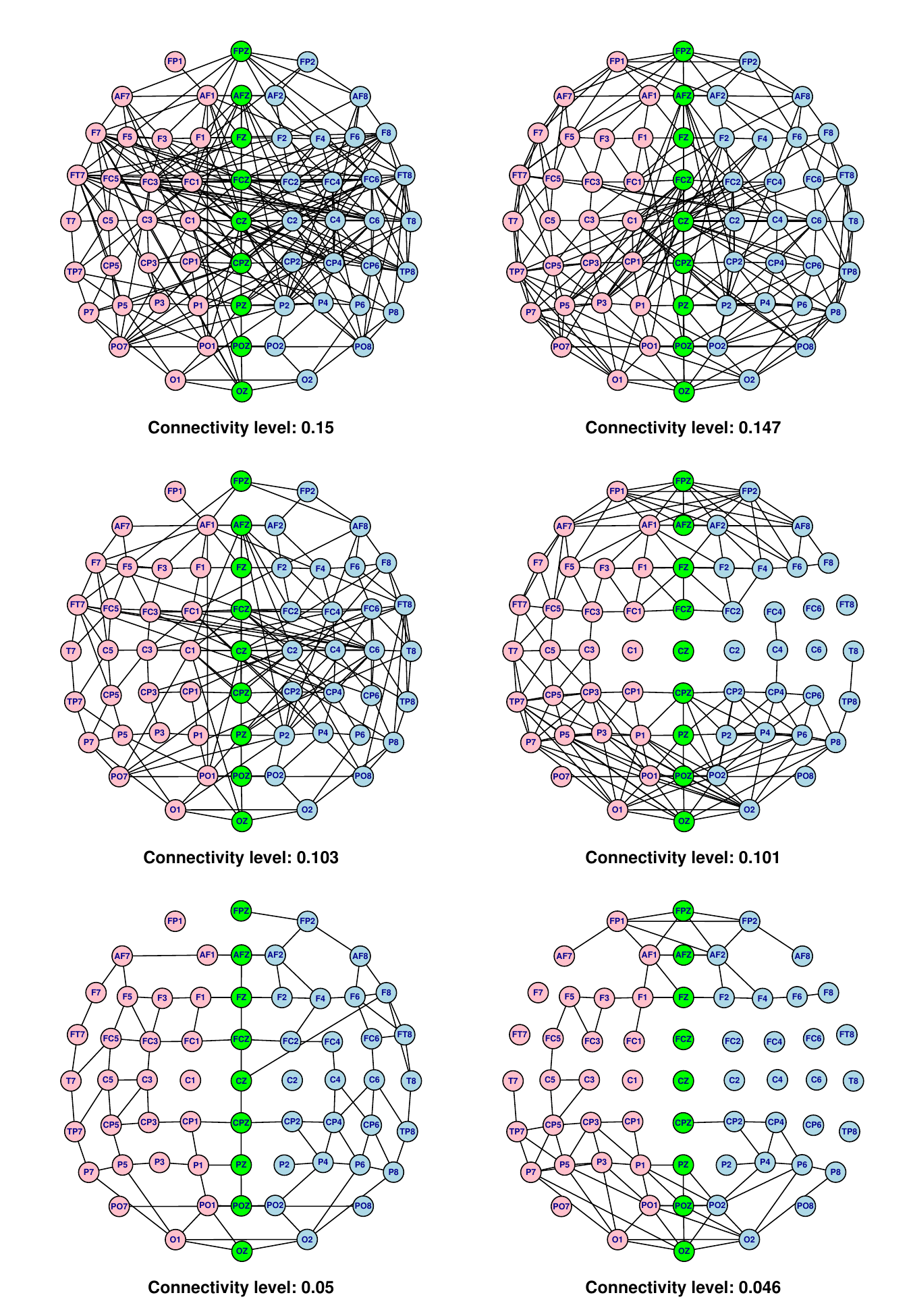}
    \caption{Estimated graphs of electrode connections by matrixNS (left) and GEMINI (right) at connectivity levels of $0.15$ (top), $0.10$ (middle), and $0.05$ (bottom). Red, green, and blue nodes represent electrodes located on the left, middle, and right regions of the scalp, respectively.}
    \label{fig:EEGrow}
\end{figure}
\noindent
We further calculate within‐ and between‐group connectivities of the three spatial regions, which are summarized in Table~\ref{table:sparsity}.

\begin{table}[H]
	\centering
	\begin{tabular}{|c|c|c|}
		\hline
		& matrixNS & GEMINI \\ \hline
		Left - Left & 0.175    & 0.138  \\ \hline
		Middle - Middle & 0.197    & 0.152  \\ \hline
		Right - Right & 0.145    & 0.222  \\ \hline
		Left - Middle & 0.115    & 0.093  \\ \hline
		Middle - Right & 0.087    & 0.103  \\ \hline
		Left - Right & 0.033    & 0.027  \\ \hline
	\end{tabular}
	\caption{Within‐ and between‐region connectivities in the estimated row-wise graphs at the connectivity level of 0.10.}
	\label{table:sparsity}
\end{table}
First, the results show denser connection among spatially adjacent electrodes, suggesting stronger functional associations in nearby regions.
In contrast, the connectivity between the left and right regions is relatively low, likely reflecting their spatial separation.
The proposed method tends to produce graphs with more evenly distributed edges compared to GEMINI. In particular, the graph estimated by GEMINI contains more isolated nodes (see the bottom of Figure \ref{fig:EEGrow}), i.e. nodes not connected to any others. 
In EEG-based analysis, it is often reported that volume conduction -
the passive spread of electrical currents through the brain, skull,
and scalp
- often leads to the observation of densely connected networks, as a
single neural source can be simultaneously detected by multiple
electrodes (\citet{Nolte:2004}).  This may explain why the proposed method yields networks with evenly and densely distributed edges.
On the other hand, we find that the graph estimated by GEMINI contains
a larger number of isolated nodes. However, biologically,
truly isolated cortical regions (nodes) are rarely observed.  For example, it is reported that nearly all cortical areas participate in
frequency-specific oscillatory coupling with others, and this
pervasive connectivity explains that functional isolation is rare in the
healthy
brain (\citet{Hipp:2012}). This suggests that the disconnected nodes
identified by GEMINI may be attributable, at least in part, to
methodological factors rather than true neurophysiological isolation.


\begin{figure}[H]
    \centering
    \includegraphics[width=0.4\textwidth]{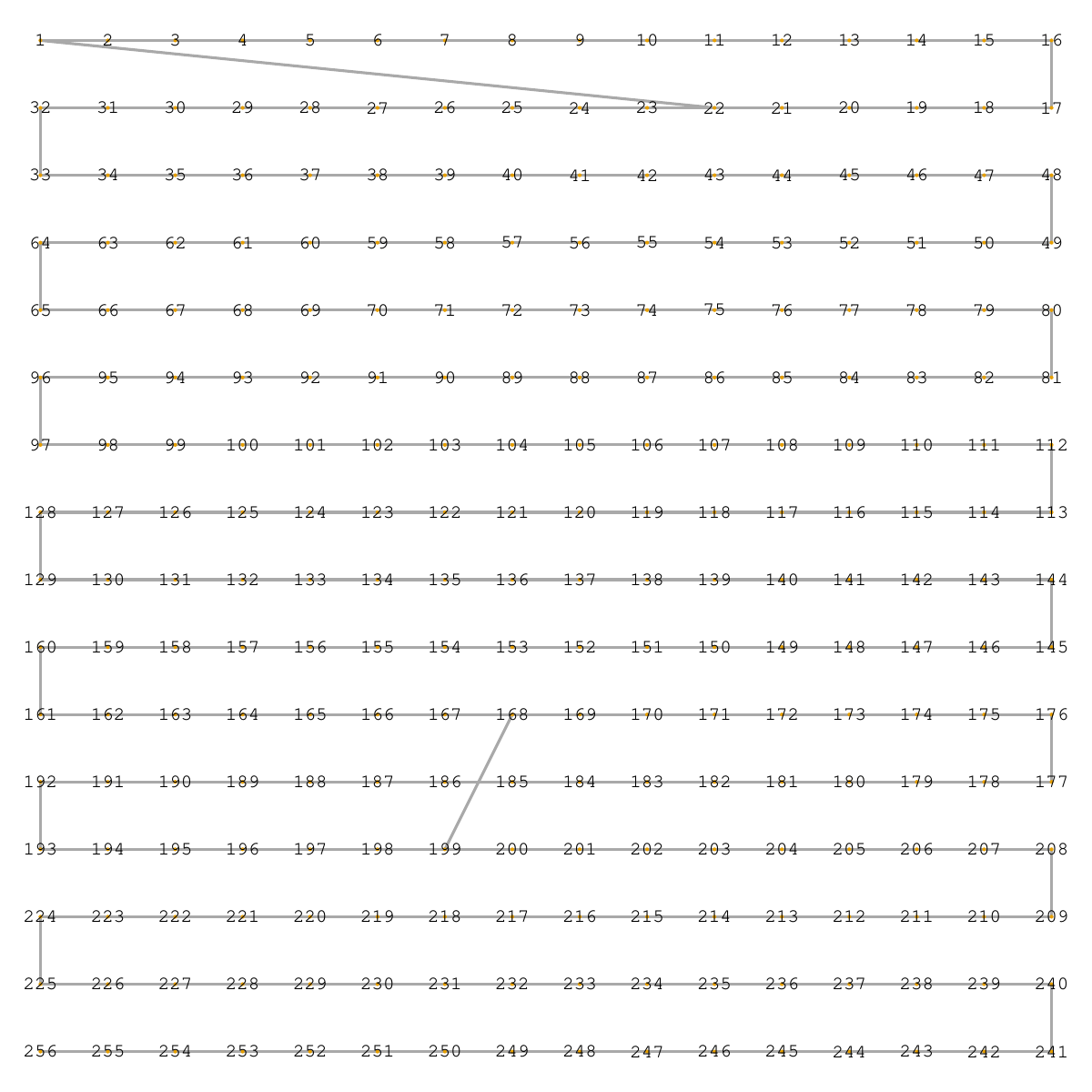}
    \caption{The estimated column-wise graph of temporal measurements using matrixNS with $\lambda = 4.0$. Nodes ($1–256$) correspond to individual measurement epochs.}
    \label{fig:EEGcol}
\end{figure}
\noindent
Figure \ref{fig:EEGcol} visualizes the column-wise graph estimated by matrixNS at $\lambda=4.0$. The network shows mostly contiguous connections across the $q=256$ time points, suggesting that the proposed method successfully captures the underlying temporal correlation structure.

\section{Discussion}\label{sec:discussion}
This paper addresses the problem of estimating the sparse structure of an undirected graph from matrix-variate data. Leveraging the Kronecker product structure of the covariance matrix, we separately estimate the row-wise and column-wise graphs, identifying each node’s neighborhood via $\ell_1$‐penalized regression. Under mild regularity and incoherence conditions, we prove that the estimated graph accurately recovers the true support of edges with high probability. We further demonstrate the practical utility of the proposed method through finite‐sample simulations and an EEG dataset.

Although we have focused on sparse precision matrices for both row and column directions, our approach applies even when the two precision factors have different structures; for example, one may be sparse while the other is low‐rank. 
Indeed, the proposed method can recover the sparsity pattern of one matrix regardless of the structure of the other since each precision component is estimated independently, as shown in \eqref{eq:lassomodel}. Moreover, Theorem \ref{mainthm} depends only on bounds for the largest and average eigenvalues of the non‐target matrix $\bV$. Given that these eigenvalue conditions are commonly used in the literature, it suggests that our theoretical guarantee can be extended to more general separable covariance structures. We believe that further exploration in this direction could yield valuable insights.

Unlike likelihood-based methods, our proposed framework does not provide an explicit estimator for the full precision matrix, as the regression coefficient is only proportional to a column of the inverse covariance matrix (see \eqref{eq:relation_coef_precision}). While our primary focus is on support recovery, we acknowledge this limitation. This issue can be addressed by jointly estimating partial variances along with partial correlations, as in \citet{Peng:2009,Khare:2014}. We leave this extension for future work.

Lastly, our estimation procedure does not rely on the normality assumption, as it employs least squares on the row and column vectors. However, the assumption plays a role in deriving the concentration inequality in Lemma \ref{lem:Gamma_SS_dev}, where we utilize the independence of decorrelated Gaussian components. We believe that this assumption could be relaxed to a sub-Gaussian assumption, such as Condition (A) in \citet{Leng:2018}.
Moreover, since we repeatedly use the conditional distribution of a multivariate Gaussian vectors, an extension to the elliptical distribution family may also be worthwhile.

\section*{Acknowledgment}
Johan Lim was supported by the government of the Republic of Korea (MSIT) and the National Research Foundation of Korea (RS-2025-00520739).
Seongoh Park was supported by the government of the Republic of Korea (MSIT) and the National Research Foundation of Korea (RS-2024-00338876).

\bibliographystyle{apalike}
\bibliography{references}

\newpage
\appendix 
\section{Proof of Theorem \ref{mainthm}}\label{sec:pf_mainthm} 
\begin{proof}
    Fix $a \in [p]$. We have $n$ i.i.d copies of matrix-valued data $\bX^{(1)}, \ldots, \bX^{(n)} \sim \mathcal{MN}_{p, q}(\bzero, \bU,\bV)$, and write
    $(\bX^{(i)})^\top = \left(X^{(i)}_{a}, \bX^{(i)}_{(a)}\right)$.
    Let $\bbX=(\bX^{(1)}, \ldots, \bX^{(n)})^\top = (\bbX_a, \bbX_{(a)}) \in \mathbb{R}^{nq \times p}$ be the vectorized data matrix, where $\bbX_a$ is the $nq \times 1$ vector obtained by stacking the $a$-th rows of each matrix data, and $\bbX_{(a)}$ is the $nq \times (p-1)$ matrix obtained by removing those rows and stack. For notational convenience, set $S=\mathcal{N}(a)$, $\bGamma=\text{tr}(\bV)/q \cdot \bU_{(a),(a)}$, and $\hat{\bGamma}=\bbX_{(a)}^{\top}\bbX_{(a)}/nq$ so that $\mathbb{E}[\hat{\bGamma}]=\bGamma$. As we noted in Section \ref{subsection:method}, we consider the standard lasso estimator: $$\hat{\theta} \in \argmin_{\theta \in \mathbb{R}^{p-1}} \left\{  \frac{1}{2n} \Vert \bbX_a - \bbX_{(a)}\theta \Vert_2^2 +\lambda \Vert \theta \Vert_1\right\}.$$

\subsection{Proof of Assertion (i) of Theorem \ref{mainthm}}
To prove (i), we employ the primal-dual witness argument of \citet{wainwright2019high}. Specifically, we seek an optimal primal-dual pair $(\hat{\theta}, \hat{z})$ that satisfies $\bbX_{(a)}^\top (\bbX_{(a)} \hat{\theta}-\bbX_a)/n+\lambda\hat{z} = 0$. 
A straightforward calculation shows 
\begin{align*}
\hat{z}_{S^c} = \underbrace{\bbX^\top_{S^c}\bbX_S(\bbX^\top_S \bbX_S)^{-1}\hat{z}_S}_{=:\ \mu} + \underbrace{\bbX^\top_{S^c}(\text{I}-\bPi_S) \frac{\epsilon}{nq \lambda}}_{=:\ T_{S^c}}
\end{align*}
where $\hat{z}^\top = (\hat{z}_S^\top, \hat{z}_{S^c}^\top)$ and $\bPi_S=\bbX_S(\bbX_S ^\top \bbX_S)^{-1}\bbX_S^\top$.
    
We first claim that $\Vert \hat{z}_{S^c} \Vert _{\infty} < 1$ holds with high probability. Lemma 7.23 in \citet{wainwright2019high} states that if the minimum eigenvalue of $\hat{\bGamma}_{SS}=\bbX_S^\top \bbX_S/nq$ is bounded away from zero, then $\Vert \hat{z}_{S^c} \Vert _{\infty} < 1$ implies $\hat{\theta}_{S^c}=0$, which in turn gives $\widehat{\mathcal{N}}(a) \subseteq \mathcal{N}(a)$. Then we complete the proof of (i) by applying a union bound over all $a \in [p]$, as will be discussed later.
    
\subsubsection*{Probabilistic Bound of $\mu$}
Since $(\bX_{S,j}^\top, \bX_{S^c,j}^\top)^\top \sim \mathcal{N}(0,v_{jj}\bU_{(a),(a)})$ for every $j \in [q]$, we can write 
$$
\bX^{(i)}_{S^c,j}=\bGamma_{S^cS}(\bGamma_{SS})^{-1}\bX^{(i)}_{S,j} + \Tilde{W}^{(i)}_{S^c,j} \in \mathbb{R}^{p-1-s}.
$$
By stacking these vectors appropriately, we obtain $\bX^{(i)}_{S^c}=\bGamma_{S^cS}(\bGamma_{SS})^{-1}\bX^{(i)}_{S} + \Tilde{\bW}^{(i)}_{S^c}$,
and hence
\begin{align}\label{decomposition}
	\bbX^\top_{S^c}=\bGamma_{S^cS}(\bGamma_{SS})^{-1}\bbX^\top_S + \Tilde{\bW}^\top_{S^c},
\end{align}
where $\Tilde{\bW}_{S^c} \in \mathbb{R}^{nq \times (p-1-s)}$ is a zero-mean Gaussian random matrix independent of $\bbX_S$. Substituting \eqref{decomposition} gives
\begin{align*}      
    \Vert \mu \Vert _{\infty}&\leq  \Vert \bGamma_{S^cS}(\bGamma_{SS})^{-1}\hat{z}_S \Vert _{\infty}+ \left\lVert \frac{\Tilde{\bW}^\top_{S^c}\bbX_S}{nq}(\hat{\bGamma}_{SS})^{-1}\hat{z}_S \right\rVert _{\infty}\\
	&\leq (1-\alpha) + \Biggl\lVert  \underbrace{\frac{\Tilde{\bW}^\top_{S^c}\bbX_S}{nq}(\hat{\bGamma}_{SS})^{-1}\hat{z}_S}_{=:R} \Biggr\rVert_{\infty}
 \end{align*}
by Assumption \ref{assum:alpha} (recall the definition of $\bGamma$).

Since $\Tilde{\bW}^\top_{S^c} \indep \bbX_S$, the conditioned $R | \bbX_S$ follows the mean zero Gaussian distribution.
Let $r_k=e_k^\top R$ denote the $k$-th entry of $R$ for each $k \in S^c$, and define $b_S:=\bbX_S(\hat{\bGamma}_{SS})^{-1} \hat{z}_S / \sqrt{nq}$. 
Then since $r_k=e_k^\top R= \Tilde{W}_k^\top b_S/\sqrt{nq}$, we have
\begin{align*}
	\text{Var}(r_k | \bbX_S)=\frac{1}{nq}\text{Var}(b_S^{\top} \Tilde{W}_k)=\frac{1}{nq}b_S^{\top} \text{Var}(\Tilde{W}_k) b_S
	\leq \frac{1}{nq} \Vert b_S \Vert _2^2  \Vert \text{Var}(\Tilde{W}_k) \Vert _2.
\end{align*}
Now observe that $\Tilde{W}_k = (\Tilde{W}_{k,1}^{(1)}, \ldots, \Tilde{W}_{k,q}^{(1)}, \ldots, \Tilde{W}_{k,1}^{(n)}, \ldots, \Tilde{W}_{k,q}^{(n)})^\top$. Also, one can verify that 
$$
\text{Cov}(\Tilde{W}_{k,l}^{(i)}, \Tilde{W}_{k,m}^{(j)})=\delta_{ij}(u_{kk}-\bU_{kS}\bU_{SS}^{-1}\bU_{Sk})v_{lm}
$$ 
where $\delta_{ij}$ denotes Kronecker delta. Combining these facts, we obtain 
$$
\text{Var}(\Tilde{W}_k)=(u_{kk}-\bU_{kS}\bU_{SS}^{-1}\bU_{Sk})\cdot \diag_{n}(\bV).
$$
Since $\bU_{SS}^{-1}$ is positive definite, it follows $u_{kk}-\bU_{kS}\bU_{SS}^{-1}\bU_{Sk} \leq u_{kk}$. Moreover, $\bV$ and $\diag_n(\bV)$ share the same eigenvalues. Taken together, we conclude that $ \Vert \text{Var}(\Tilde{W}_k) \Vert _2 \leq  u_{kk} \Vert \bV \Vert _2$.

Now the uniform upper bound for the conditional variances holds for all $k \in [p]$:
\begin{align}
	\sqrt{\text{Var}(r_k|\bbX_S)} &= \sqrt{\frac{ u_{kk} \Vert \bV \Vert _2}{nq}} \left\lVert \frac{\bbX_S}{\sqrt{nq}}\hat{\bGamma}_{SS}^{-1}\hat{z}_S \right\rVert _2 \nonumber\\
	&\leq \sqrt{\frac{ u_{kk} \Vert \bV \Vert _2}{nq}} \left\lVert \frac{\bbX_S}{\sqrt{nq}}\hat{\bGamma}_{SS}^{-1} \right\rVert _2\cdot  \Vert \hat{z}_S \Vert _2 \nonumber\\
	&\leq \sqrt{\frac{ u_{kk} \Vert \bV \Vert _2}{nq}}\sqrt{ \Vert \hat{\bGamma}_{SS}^{-1} \Vert _2}\cdot \sqrt{d_{\max}} \label{vb3}\\
	&\leq \sqrt{\frac{2d_{\max}}{nq \lambda_{\min}^U v_{\text{avg}}}\cdot u_{\max} \Vert \bV \Vert _2}. \label{vb4}
\end{align}
Inequality \eqref{vb3} follows from Assumption \ref{assum:sparse}, while \eqref{vb4} holds with high probability by the following fact:
\begin{fact}\label{fact:minimumeigenvalue}
    Let $E_{1a} : = \{ \Vert \hat{\bGamma}_{SS} - {\bGamma}_{SS} \Vert _2 \le \lambda_{\min}^Uv_{\text{avg}}/2\}$. On $E_{1a}$, we have
    \begin{align}\label{eq:minimumeigenvalue}
        \Vert \hat{\bGamma}_{SS}^{-1} \Vert _2 \leq \frac{2}{\lambda_{\min}^Uv_{\text{avg}}}.
    \end{align}
    Moreover, $E_{1a}$ holds with probability at least $1-3d_{\max}/p^\beta$ for any $beta>0$, under either condition \eqref{eq:dim_cond1} or \eqref{eq:dim_cond2}.
\end{fact}
\noindent
The proof of Fact~\ref{fact:minimumeigenvalue} is deferred to Section \ref{sec:pf_minimumeigenvalue}. As a consequence, by applying the Gaussian tail property and a union bound, we have
\begin{align*}
    &\mathbb{P}( \Vert R \Vert _{\infty} \geq t) 
    \leq \mathbb{P}(\max_{k \in S^c} \vert r_k \vert  \geq t |\ E_{1a}) + \mathbb{P}(E_{1a}^c)\\
    &\leq 2|S^c|\exp \left( -\frac{nq \lambda_{\min}^U v_{\text{avg}} t^2}{4d_{\max} \cdot u_{\max} \Vert \bV \Vert _2} \right) +\mathbb{P}(E_{1a}^c)\\
    &\leq 2p \exp \left(-\frac{nq\lambda_{\min}^U v_{\text{avg}} t^2}{4d_{\max}\cdot u_{\max} \Vert \bV \Vert _2}\right)+\mathbb{P}(E_{1a}^c).
\end{align*}
For any fixed $\beta > 1$, by taking $t=\alpha/4$, we have
\begin{align*} 
    \mathbb{P}\left( \Vert R \Vert _{\infty} \geq \frac{\alpha}{4} \ |\ E_{1a}\right) \leq 2p \exp \left( -\frac{\lambda_{\min}^Uv_{\text{avg}} nq \alpha^2}{64d_{\max}\cdot u_{\max} \Vert \bV \Vert _2 }\right)  \leq \frac{2}{p^{\beta}},
\end{align*}
provided that 
\begin{align} \label{eq:dim1}
    \frac{nq}{\log p} \geq (1+\beta)\frac{64d_{\max} \cdot u_{\max} \Vert \bV \Vert _2}{\lambda_{\min}^Uv_{\text{avg}} \alpha^2}.
\end{align}
Combining everything, we conclude
\begin{align}\nonumber
	\mathbb{P}\left( \Vert \mu \Vert _{\infty} \geq 1-\frac{3}{4}\alpha \right) \leq\  \mathbb{P}\left( \Vert R \Vert _{\infty} \geq \frac{\alpha}{4}\right) 
    \leq \frac{3d_{\max}+2}{p^{\beta}},
\end{align}
as long as either (i) \eqref{eq:dim1} and \eqref{eq:dim_cond1}, or (ii) \eqref{eq:dim1} and \eqref{eq:dim_cond2} are satisfied.



\subsubsection*{Probabilistic Bound of $T_{S^c}$}

Recall that $T_{S^c}=\bbX^\top_{S^c}(\text{I}-\bPi_S)\epsilon/(nq \lambda)$. Using the decomposition \eqref{decomposition} and the fact $\bbX_S^\top(\text{I}-\bPi_S)=0$, we have $T_{S^c}=\Tilde{\bW}_{S^c}^\top(\text{I}-\bPi_S)\epsilon/(nq \lambda)$.
To apply the same arguments leading to \eqref{vb4}, we need a probabilistic bound on $ \Vert (\text{I}-\bPi_S)\epsilon \Vert _2$. Since $\text{I}-\bPi_S$ is a projection, we get $\Vert (\text{I}-\bPi_S)\epsilon \Vert _2 \leq  \Vert \epsilon \Vert _2$. Therefore, it remains to bound $\mathbb{P}( \Vert \epsilon \Vert _2 \geq t)$.

By the model assumption, for each $i \in [n], j \in [q]$, we have $X^{(i)}_{a,j}=X^{(i)}_{(a),j}\theta^*+\epsilon^{(i)}_j$. Since $X^{(i)}_j \sim \mathcal{N}_q(0,v_{jj}\bU)$, one can check $\text{Var}(\epsilon^{(i)}_j)=v_{jj}(u_{aa}-\bU_{a,(a)}(\bU_{(a),(a)})^{-1}\bU_{(a),a})$, or equivalently
$$
\text{Var}(\epsilon^{(i)})=(u_{aa}-\bU_{a,(a)}(\bU_{(a),(a)})^{-1}\bU_{(a),a})\bV =: \bV^{*}.
$$
Thus we can write $\epsilon^{(i)}\overset{d}{=}{\bV^*}^{1/2}Z^{(i)}$  with 
$Z^{(i)} \overset {\mathrm{i.i.d.}}{\sim} \mathcal{N}_q(0,I_q)$. Let
$\Tilde{Z}=(Z^{(1)},\dots,Z^{(n)})^\top$ and
$\Tilde{\bV}^*=\diag_n(\bV^*)$, then the squared norm becomes
\begin{align*}
	\Vert \epsilon \Vert _2^2 = \sum_{i=1}^n ({\epsilon^{(i)}})^\top \epsilon^{(i)}
	{\buildrel d \over =} \sum_{i=1}^n ({Z^{(i)}})^\top \bV^{*} Z^{(i)} 
	= \Tilde{Z}^\top \Tilde{\bV}^* \Tilde{Z}.
\end{align*}
By the Hanson–Wright inequality \citep{rudelson2013hanson}, for any $s>0$,
\begin{align}\label{eq:HWineq}
	\mathbb{P}( \Vert \epsilon \Vert _2^2 \geq \text{tr}(\Tilde{\bV}^*)+s) \leq&\ 2 \exp\Bigl( -c \min\bigl(\frac{s}{ \Vert \Tilde{\bV^*} \Vert _2}, \frac{s^2}{ \Vert \Tilde{\bV^*} \Vert _F^2}\bigr)\Bigr)  \\
    \leq &\ 2 \exp\Bigl( -c \min\bigl(\frac{s}{ u_{aa} \Vert \bV \Vert _2}, \frac{s^2}{n q u_{aa}^2  \Vert \bV \Vert_2^2}\bigr)\Bigr), \nonumber
\end{align}
where we used $u_{aa}-\bU_{a,(a)}(\bU_{(a),(a)})^{-1}\bU_{(a),a} \leq u_{aa}$ and $\Vert \bV \Vert _F^2 \leq q \Vert \bV \Vert_2^2$.
Setting $s=nq u_{aa} v_{\text{avg}}$ and since $\text{tr}(\Tilde{\bV^*}) \leq nu_{aa}\text{tr}(\bV) = nq u_{aa} v_{\text{avg}}$ and 
$v_{\text{avg}} \leq \Vert \bV \Vert_2$, we get 
$$
\mathbb{P}\bigl(\|\epsilon\|_2 \ge \sqrt{2nq\,u_{aa}\,v_{\text{avg}}} \bigr)
 \leq\ 2 \exp\Bigl( -c \cdot nq \min\left(\frac{v_{\text{avg}}}{ \Vert \bV \Vert _2}, \frac{v_{\text{avg}}^2}{ \Vert \bV \Vert _2^2}\right)\Bigr)\\
\le 2 \exp\Bigl( -c \cdot nq \frac{v_{\text{avg}}^2}{ \Vert \bV \Vert _2^2} \Bigr).
$$
Therefore, the event $E_{2a}:=\{ \Vert \epsilon \Vert _2 \leq \sqrt{2nq\, u_{aa}\, v_{\text{avg}}}\}$ satisfies $\mathbb{P}(E_{2a}^c) \leq 2/p^{\beta}$ whenever
\begin{align} \label{eq:dim2}
    \frac{nq}{\log p} \geq \frac{\beta\, \Vert \bV \Vert _2^2}{c\, v_{\text{avg}}^2}.
\end{align}

Since $\Tilde{\bW}_{S^c}^\top$ and $\epsilon$ are independent, $T_{S^c}|\epsilon$ follows the centered normal distribution. Denote its $k$\-th component by $t_k := e_k^\top T_{S^c}=\Tilde{W}_k^\top(\text{I}-\bPi_S)\epsilon/(nq\lambda)$ . On $E_{2a}$, its conditional variance satisfies
\begin{align*}
	\text{Var}(t_k|\epsilon)=&\ \frac{1}{(nq\lambda)^2}((\text{I}-\bPi_S)\epsilon)^\top \text{Var}(\Tilde{W}_k) (\text{I}-\bPi_S) \epsilon
	\leq\ \frac{1}{(nq\lambda)^2} \Vert (\text{I}-\bPi_S)\epsilon \Vert _2^2  \Vert \text{Var}(\Tilde{W}_k) \Vert _2\\
	\leq&\ \frac{1}{(nq\lambda)^2} \Vert \epsilon \Vert _2^2 u_{kk} \Vert \bV \Vert _2
	\leq\ \frac{2}{nq\lambda^2} u_{aa}u_{kk} v_{\text{avg}} \Vert \bV \Vert _2.
\end{align*}
By the Gaussian tail bound and a union bound, we obtain
$$
\mathbb{P}( \Vert T_{S^c} \Vert _{\infty} \geq t) \leq \mathbb{P} \bigl(\max_{k \in S^c} \vert t_k \vert  \geq t \ |\ E_{2a} \bigr) + \mathbb{P}(E_{2a}^c) \leq 2|S^c|\exp \left(-\frac{nq\lambda^2 t^2}{4 u_{\max}^2 v_{\text{avg}} \Vert \bV \Vert _2} \right)+\mathbb{P}(E_{2a}^c).
$$
Setting $t=\alpha/4$ , we get
\begin{align*}
    \mathbb{P}\left( \Vert T_{S^c} \Vert _{\infty} \geq \frac{\alpha}{4} \right) \leq 2p \exp \left(-\frac{nq\lambda^2 \alpha^2}{64 u_{\max}^2 v_{\text{avg}} \Vert \bV \Vert _2} \right) +\mathbb{P}(E_{2a}^c).
\end{align*}
Finally, if \eqref{eq:dim2} and 
\begin{align} \label{lambdacond1}
        \lambda \geq \frac{1}{\alpha}\sqrt{\frac{\log p}{nq}(1+\beta)\cdot 64u_{\max}^2 v_{\text{avg}} \Vert \bV \Vert _2},
\end{align}
holds, we get $\mathbb{P}( \Vert T_{S^c} \Vert _{\infty} \geq \alpha/4) \le 4/p^\beta$.

\subsubsection*{Combining Pieces Together}
Putting everything together, we deduce that for any $\beta>1$,
\begin{align}\label{eq:zbound}
        \mathbb{P}\left( \Vert \hat{z}_{S^c} \Vert_\infty \geq 1-\frac{\alpha}{2}\right) \leq \frac{3d_{\max}+6}{p^\beta}
\end{align}
provided that either (i) \eqref{eq:dim1} and \eqref{eq:dim_cond1} or (ii) \eqref{eq:dim1} and \eqref{eq:dim_cond2} hold together with \eqref{lambdacond1}.
As we have already proved that the minimum eigenvalue of $\hat{\bGamma}_{SS}$ is bounded away from zero on $E_{1a}$ in \eqref{eq:minimumeigenvalue}, it follows that
\begin{align*}
        \mathbb{P}\left( \widehat{\mathcal{N}}(a) \subseteq \mathcal{N}(a) \right) \geq 1-\frac{3d_{\max}+6}{p^\beta},
\end{align*}
under the same assumptions. 
Finally, taking a union bound over $a=1, \ldots, p$, we have $\hat{E} \subseteq E$ with probability larger than $1-(3d_{\max}+6)/p^{\beta-1}$, which completes the proof.

\subsubsection{Proof of Fact \ref{fact:minimumeigenvalue}}\label{sec:pf_minimumeigenvalue}
First, observe that
\begin{align} \label{b_bound}
    \Vert {\bGamma}_{SS}^{-1} \Vert _2 =\ \lambda_{\min}({\bGamma}_{SS})^{-1}
    =\ \frac{q}{\text{tr}(\bV)}\cdot\lambda_{\min}(\bU_{SS})^{-1}
    \leq \frac{1}{v_{\text{avg}} \lambda_{\min}^U}
\end{align}
Next, by using Weyl's inequality, we get
$$
\left| \frac{1}{ \Vert \hat{\bGamma}_{SS}^{-1} \Vert _2} - \frac{1}{ \Vert {\bGamma}_{SS}^{-1} \Vert _2} \right| =|\lambda_{\min}(\hat{\bGamma}_{SS}) - \lambda_{\min}({\bGamma}_{SS})| \le  \Vert \hat{\bGamma}_{SS} - {\bGamma}_{SS} \Vert _2.
$$
Thus, on the event $E_{1a} = \{ \Vert \hat{\bGamma}_{SS} - {\bGamma}_{SS} \Vert _2 \le \lambda_{\min}^Uv_{\text{avg}}/2\}$, we get
$$
\frac{1}{ \Vert \hat{\bGamma}_{SS}^{-1} \Vert _2} \ge \frac{1}{ \Vert {\bGamma}_{SS}^{-1} \Vert _2}  - \Vert \hat{\bGamma}_{SS} - {\bGamma}_{SS} \Vert _2 \ge \frac{\lambda_{\min}^Uv_{\text{avg}}}{2},
$$
which proves \eqref{eq:minimumeigenvalue}.

We now show $E_{1a}^c$ holds with high probability. By Lemma \ref{lem:Gamma_SS_dev} with $t_0 = \lambda_{\min}^U v_{\text{avg}}/(2C ||\bU||_2)$, one obtains
$$
{\rm P} \left[
\Vert \hat{\bGamma}_{SS} - {\bGamma}_{SS} \Vert _2 
\geq \lambda_{\min}^Uv_{\text{avg}} / 2 \right] \le 3d_{\max} \exp \bigg[
- \min \Big\{
(t_0 / a_1)^2, t_0 / a_2, (t_0 / a_3)^{2/3}
\Big\}
\bigg].
$$
where $a_1, a_2, a_3$ are as in Lemma \ref{lem:Gamma_SS_dev}. The right hand side is at most $3d_{\max}/p^\beta$ if one of the following holds:
\begin{enumerate}[label=(\roman*)]
	\item[(i)] $1/a_1 \ge (\beta \log p)^{1/2}/t_0$ and $t_0 \le \min\{ a_1^2/a_2, a_1^{3/2} / a_3^{1/2}\}$,
	\item[(ii)] $1/a_2 \ge \beta \log p/t_0$ and $ a_1^2/a_2 \le t_0 \le a_2^3/a_3^2$,
	\item[(iii)] $1/a_3 \ge (\beta \log p)^{3/2}/t_0$ and $t_0 \ge \max\{a_1^{3/2} / a_3^{1/2},  a_2^3/a_3^2\}$.
\end{enumerate}
Under the mild assumption $n (2d_{\max}/q + 1) > ||\bV||_2^2$, then all the three terms 
$$ \frac{a_1^{3/2} / a_3^{1/2}}{a_2^3/a_3^2} =  \left(\dfrac{n (2d_{\max}/q + 1)}{||\bV||_2^2} \right)^{3/4}, \ 
	\frac{a_1^2/a_2}{a_1^{3/2} / a_3^{1/2}} = \left(\dfrac{n (2d_{\max}/q + 1)}{||\bV||_2^2} \right)^{1/4}, 
	\ \frac{a_1^2/a_2}{a_2^3/a_3^2} = \dfrac{n (2d_{\max}/q + 1)}{||\bV||_2^2}
$$
are greater than $1$. Consequently, the conditions (i)-(iii) reduce to:
\begin{enumerate}[label=(\alph*)]
	\item[(a)] 
	$1/a_1 \ge (\beta \log p)^{1/2}/t_0$ and $t_0 \le a_1^{3/2} / a_3^{1/2}$,
	\item[(b)] $1/a_3 \ge (\beta \log p)^{3/2}/t_0$ and $t_0 \ge a_1^{3/2} / a_3^{1/2}$.
\end{enumerate}
The sufficient condition of (a) is that 
\begin{equation}\label{eq:dim_cond1}
	n\! \left( \frac{2d_{\max}}{q}\! + 1 \right)\! >\! ||\bV||_2^2, \ \ \frac{n}{\log p} \ge  \beta \!\left(\frac{2d_{\max}}{q}  + 1 \right)\!  \frac{||\bV||_2^2}{t_0^2}, \ \ \frac{q^2}{n} \ge \left( \frac{2d_{\max}}{q}\! +1 \right)^{-3}\! \frac{t_0^4}{||\bV||_2^2}.
\end{equation}
Similarly for the condition (b), we need
\begin{equation}\label{eq:dim_cond2}
	n \left(\frac{2d_{\max}}{q}\! + 1 \right)\! > \!||\bV||_2^2, \ \ \frac{nq}{(\log p)^{3/2}} \ge \beta^{3/2} \frac{||\bV||_2^2}{t_0}, \ \ 
	\frac{n}{q^2} \ge  \left( \frac{2d_{\max}}{q}\! + 1 \right)^3\! \frac{||\bV||_2^2}{t_0^4}.
\end{equation}
In summary, for given $\beta >0$, if either \eqref{eq:dim_cond1} or \eqref{eq:dim_cond2} holds, then we have
$$
{\rm P} \left(
\Vert \hat{\bGamma}_{SS} - {\bGamma}_{SS} \Vert _2 
\geq \frac{\lambda_{\min}^Uv_{\text{avg}}}{2} \right) \le \frac{3d_{\max}}{p^\beta}.
$$

\subsection{Proof of (ii) of Theorem \ref{mainthm}}
One can easily check that it suffices to show 
$\Vert \hat{\theta}_S-\theta_S^* \Vert _{\infty} \leq 3\lambda \sqrt{d_{\max}}/(\lambda_{\min}^Uv_{\text{avg}})$ 
holds with high probability.
Indeed, we observe
\begin{align}
	\Vert \hat{\theta}_S-\theta_S^* \Vert _{\infty} \leq&\  \left\lVert \bigl(\frac{\bbX_S^\top \bbX_S}{nq}\bigr)^{-1}\bbX_S^\top \frac{\epsilon}{nq} \right\rVert _{\infty}+ \lambda \left\lVert \bigl(\frac{\bbX_S^\top\bbX}{nq}\bigr)^{-1} \right\rVert _{\infty} \nonumber\\
	=&\ \Big\Vert \underbrace{\hat{\bGamma}_{SS}^{-1}\bbX_S^\top \frac{\epsilon}{nq}}_{=:L_S} \Big\Vert _{\infty} + \lambda \Vert \hat{\bGamma}_{SS}^{-1} \Vert _{\infty}. \nonumber
\end{align}
Moreover, under the event $E_{1a} = \{ \Vert \hat{\bGamma}_{SS}-{\bGamma}_{SS} \Vert _2 \leq \lambda_{\min}^Uv_{\text{avg}}/2\}$,  we have
\begin{align}
	\Vert \hat{\bGamma}_{SS}^{-1} \Vert _{\infty} \leq&\ \Vert \bGamma_{SS}^{-1} \Vert _{\infty} + \Vert \hat{\bGamma}_{SS}^{-1}-\bGamma_{SS}^{-1} \Vert _{\infty} \nonumber \\
	\leq&\ \frac{\sqrt{d_{\max}}}{\lambda_{\min}^Uv_{\text{avg}}} + \sqrt{d_{\max}} \Vert \hat{\bGamma}_{SS}^{-1}-\bGamma_{SS}^{-1} \Vert _2 \label{eqn:infnorm2}\\
	\leq&\ \frac{\sqrt{d_{\max}}}{\lambda_{\min}^Uv_{\text{avg}}} +\sqrt{d_{\max}}\cdot\frac{ \Vert \hat{\bGamma}_{SS}-{\bGamma}_{SS} \Vert _2 /(\lambda_{\min}^Uv_{\text{avg}})^2}{1-\Vert \hat{\bGamma}_{SS}-{\bGamma}_{SS} \Vert _2 /\lambda_{\min}^Uv_{\text{avg}}} \label{eqn:infnorm3}  \\
	\leq&\ \frac{2\sqrt{d_{\max}}}{\lambda_{\min}^Uv_{\text{avg}}}. \nonumber
\end{align}
For \eqref{eqn:infnorm2}, we use the fact 
$$
\Vert \bGamma_{SS}^{-1} \Vert _{\infty} \leq \sqrt{d_{\max}} \Vert \bGamma_{SS}^{-1} \Vert _2 \leq \sqrt{d_{\max}}/(\lambda_{\min}^Uv_{\text{avg}})
$$ 
which is verified by \eqref{b_bound}, and \eqref{eqn:infnorm3} follows from Lemma \ref{lem:invcov_dev}. 

Next, we focus on the term $L_S = \hat{\bGamma}_{SS}^{-1}\bbX_S^\top \epsilon/(nq)$. Since $\hat{\bGamma}_{SS}^{-1}\bbX_S^\top \indep \epsilon$, $L_S | \bbX_S$ follows the centered normal distribution. Let $\bA= \bbX_S \hat{\bGamma}_{SS}^{-1}/nq$, and denote its $j$-th column by $a_j$. Then, the maximum conditional variances of  the components of $L_S | \bbX_S$ is bounded as
$$
\max_{j \in S} \text{Var}((L_S)_j| \bbX_S) = \max_{j \in S}a_j^\top \text{Var}(\epsilon) a_j \le  \Vert \bV \Vert _2 \max_{j} a_j^\top a_j \le  \Vert \bV \Vert _2  \Vert  \diag(\bA^\top \bA)  \Vert 
$$
since $\text{Var}(\epsilon) = \text{diag}_n(\bV) \preceq  \Vert \bV \Vert _2 \text{I}_{nq}$.
In the last inequality, any matrix norm can be used, such as $ \Vert \cdot \Vert _{\max}$, $ \Vert \cdot \Vert _{2}$, and $ \Vert \cdot \Vert _{\infty}$. Finally, as $\bA^\top \bA = \widehat{\bGamma}_{SS}^{-1} / nq$, we can conclude under $E_{1a}$ that
\begin{equation}
	\max_{j \in S} \text{Var}((L_S)_j| \bbX_S) \le \dfrac{ \Vert \bV \Vert _2}{nq}  \Vert  \diag(\widehat{\bGamma}_{SS}^{-1})  \Vert_{\infty}  \leq \frac{ \Vert \bV \Vert _2}{nq} \Vert \widehat{\bGamma}_{SS}^{-1}  \Vert _{\infty} \leq \frac{2 \sqrt{d_{\max}} \Vert \bV \Vert _2}{nq\cdot \lambda_{\min}^Uv_{\text{avg}}} \nonumber
\end{equation}
where the last inequality uses $\Vert \hat\bGamma_{SS}^{-1} \Vert _{\infty} \leq \sqrt{d_{\max}} \Vert \hat\bGamma_{SS}^{-1} \Vert _2 \leq 2\sqrt{d_{\max}}/ (\lambda_{\min}^Uv_{\text{avg}})$. Applying the Gaussian tail bound,we get
\begin{align}\label{lambda_tailprob}
	\mathbb{P}\left( \Vert L_S \Vert _{\infty} \geq \frac{\sqrt{d_{\max}}}{\lambda_{\min}^Uv_{\text{avg}}} \lambda \right) \leq &\ 2|S| \exp \left(-nq\frac{\sqrt{d_{\max}} \lambda^2}{4 \lambda_{\min}^Uv_{\text{avg}} \Vert \bV \Vert _2} \right).
\end{align}
Hence, this probability is at most $2d_{\max} / p^\beta$ if $\lambda$ satisfies
\begin{align}\label{lambdacond2}
    \lambda \geq 2\sqrt{\frac{\log p}{nq} \cdot\beta \frac{\lambda_{\min}^Uv_{\text{avg}} \Vert \bV \Vert_2}{\sqrt{d_{\max}}} }.
\end{align}
In conclusion,under conditions \eqref{eq:dim_cond1} or \eqref{eq:dim_cond2} and \eqref{lambdacond2}, we obtain
\begin{align}\nonumber
    \mathbb{P}\left(\Vert \hat{\theta}_S-\theta_S^* \Vert _{\infty} \leq \frac{3\sqrt{d_{\max}}}{\lambda_{\min}^Uv_{\text{avg}}} \lambda \right) \leq \frac{2d_{\max}}{p^{\beta}}+\mathbb{P}(E_{1a}^c) \leq \frac{5d_{\max}}{p^\beta}
\end{align}
and the proof is completed by applying a union bound over the $p$ rows.

\end{proof}

\section{Supporting lemmas}

\subsection{Lemma \ref{lem:Gamma_SS_dev}}
\begin{lemma}\label{lem:Gamma_SS_dev}
Let $\bX^{(i)}$, $i=1,\dots,n$, be i.i.d. samples from the matrix normal distribution $\mathcal{MN}_{p,q}(\mathbf{0},\bU,\bV)$, and fix any index set $S \subseteq [p]$ with $|S| \leq m$ for some constant $m$. Let $\bGamma_{SS} = \text{tr}(\bV)/q \cdot \bU_{SS}$, $\hat{\bGamma}_{SS} = \sum_{i=1}^{n} \bX_S^{(i)} \bX_S^{(i) \top} /nq$.
Under Assumptions \ref{assum:sparse} and \ref{assum:eigen}, there exists a numerical constant $C>0$ such that for any $t>0$, 
$$
{\rm P} \left[
\Vert \hat{\bGamma}_{SS} - {\bGamma}_{SS} \Vert _2 
\geq C \cdot  \Vert \bU \Vert _2 \cdot t\right]
\le 
3m \exp \bigg[
- \min \Big\{
(t / a_1)^2, t / a_2, (t / a_3)^{2/3}
\Big\}
\bigg],
$$
where
\begin{align*}
	a_1 & = \sqrt{(2 m / q + 1) \Vert \bV \Vert^2_2 / n},\\
	a_2 & = \Vert \bV \Vert_2^2 / (n\sqrt{q}),\\
	a_3 & = \Vert \bV \Vert_2^2 / (nq).	
\end{align*}
\end{lemma}

\begin{proof}

Since $\bX_S^{(i)} \sim \mathcal{MN}_{s, q}(\bzero,  \bU_{SS}, \bV)$, write $\bX_S^{(i)}=\bU_{SS}^{1/2}\bZ_{S}^{(i)}\bV^{1/2}$ where each $\bZ_{S}^{(i)}$ has i.i.d standard normal entries. Then
$$
\hat{\bGamma}_{SS}
=\frac{1}{nq}\sum_{i=1}^n \bX_{S}^{(i)} \bigl( \bX_{S}^{(i)} \bigr)^\top 
=\frac{1}{nq}\bU_{SS}^{\frac{1}{2}}\sum_{i=1}^n \bZ_{S}^{(i)} \bV \big( \bZ_{S}^{(i)} \big)^\top \bU_{SS}^{\frac{1}{2}},
$$
thus
\begin{align*}
     \left\Vert \hat{\bGamma}_{SS} - {\bGamma}_{SS} \right\Vert _2 =&  \left\Vert \frac{1}{q} \bU_{SS}^{\frac{1}{2}} \left( \frac{1}{n} \sum_{i=1}^n \bZ_{S}^{(i)} \bV (\bZ_{S}^{(i)})^\top - \text{tr}(\bV)\text{I}_s \right) \bU_{SS}^{\frac{1}{2}} \right\Vert _2\\
    \leq&\ \frac{1}{q} \left\Vert \bU_{SS} \right\Vert _2 \cdot  \left\Vert \frac{1}{n} \sum_{i=1}^n \bZ_{S}^{(i)} \bV (\bZ_{S}^{(i)})^\top - \text{tr}(\bV)\text{I}_s \right\Vert _2\\
    \leq&\ \frac{1}{q} \left\Vert \bU \right\Vert _2 \cdot  \left\Vert \frac{1}{n} \sum_{i=1}^n \bZ_{S}^{(i)} \bV (\bZ_{S}^{(i)})^\top - \text{tr}(\bV)\text{I}_s \right\Vert _2.
\end{align*}

We apply the Bernstein inequality for bounded random matrices (Theorem 6.1.1 of \citet{tropp2015introduction}) to control the rate of $ \Vert \frac{1}{n} \sum_{i=1}^n \bZ_{S}^{(i)} \bV (\bZ_{S}^{(i)})^\top - \text{tr}(\bV)\text{I}_s \Vert _2$.
To this end, define the event $A = \cap_{i=1}^n A_i \equiv \cap_{i=1}^n \{ \Vert \bZ_S^{(i)} \Vert _2^2 \le K \}$ for some constant $K>0$ to be specified later. On the event $A$, it holds
\begin{align*}
\frac{1}{n} \sum_{i=1}^n \big\{\bZ^{(i)}_{S} \bV (\bZ^{(i)}_{S})^\top - \mathbb{E}\bZ^{(i)}_{S} \bV (\bZ^{(i)}_{S})^\top \big\}
=&\ \frac{1}{n}\sum_{i=1}^n \big\{\bZ^{(i)}_{S} \bV (\bZ^{(i)}_{S})^\top \text{I}_{A_i} - \mathbb{E}\bZ^{(i)}_{S} \bV (\bZ^{(i)}_{S})^\top \text{I}_{A_i} \big\} \\
-&\ \mathbb{E}\bZ^{(1)}_{S} \bV (\bZ^{(1)}_{S})^\top \text{I}_{A_1^c}, 
\end{align*}
which implies
\begin{align*}
    \bigg\|\frac{1}{n}\sum_{i=1}^n \bigl\{\bZ^{(i)}_{S} \bV (\bZ^{(i)}_{S})^\top - \mathbb{E}\bZ^{(i)}_{S} \bV (\bZ^{(i)}_{S})^\top \bigr\}\bigg\|_2
    &\le 
    \underbrace{\frac{1}{n}\bigg\|\sum_{i=1}^n \bigl\{\bZ^{(i)}_{S} \bV (\bZ^{(i)}_{S})^\top \text{I}_{A_i}
       - \mathbb{E}\bZ^{(i)}_{S} \bV (\bZ^{(i)}_{S})^\top \text{I}_{A_i}\bigr\}\bigg\|_2}_{=:B_1}\\
    &\quad+
    \underbrace{\bigg\|\mathbb{E}\bZ^{(1)}_{S} \bV (\bZ^{(1)}_{S})^\top \text{I}_{A_1^c}\bigg\|_2}_{=:B_2}\,.
\end{align*}

\subsubsection*{(1) Concentration of $B_1$}
We apply the matrix Bernstein inequality to $\bS_i \equiv \bZ^{(i)}_{S} \bV (\bZ^{(i)}_{S})^\top \text{I}_{A_i} - \mathbb{E}\bZ^{(i)}_{S} \bV (\bZ^{(i)}_{S})^\top \text{I}_{A_i}$.

\subsubsection*{Boundedness of $\bS_i$}
Assume the event $A$ holds. Then
\begin{align*}
    \Vert  \bS_i  \Vert _2 & \le  \Vert  \bZ^{(i)}_{S} \bV      (\bZ^{(i)}_{S})^\top  \Vert _2 +  \Vert  \mathbb{E}\bZ^{(i)}_{S} \bV (\bZ^{(i)}_{S})^\top \text{I}_{A_i}  \Vert _2 \\
    & \le  \Vert  \bZ^{(i)}_{S} \bV (\bZ^{(i)}_{S})^\top  \Vert _2 + \mathbb{E}  \Vert  \bZ^{(i)}_{S} \bV (\bZ^{(i)}_{S})^\top \text{I}_{A_i}  \Vert _2
    \le 2  \Vert \bV \Vert _2 K,   
\end{align*}
where the last inequality holds since
$\Vert  \bZ^{(i)}_{S} \bV (\bZ^{(i)}_{S})^\top  \Vert _2 \le  \Vert  \bZ^{(i)}_{S} \Vert _2^2  \Vert  \bV  \Vert _2 \le  \Vert  \bV  \Vert _2 K.$

\subsubsection*{Calculation of $ \Vert \mathbb{E} \bS_1 \bS_1^\top \Vert _2$}
Observe that
\begin{align*}
    \mathbb{E} \bS_1 \bS_1^\top & =  
    \mathbb{E} \bZ^{(1)}_{S} \bV (\bZ^{(1)}_{S})^\top \bZ^{(1)}_{S} \bV (\bZ^{(1)}_{S})^\top  \text{I}_{A_1} - \mathbb{E}\bZ^{(i)}_{S} \bV (\bZ^{(1)}_{S})^\top \text{I}_{A_1}\mathbb{E}\bZ^{(1)}_{S} \bV (\bZ^{(1)}_{S})^\top \text{I}_{A_1} \\
    & \preceq \mathbb{E} \bZ^{(1)}_{S} \bV (\bZ^{(1)}_{S})^\top \bZ^{(1)}_{S} \bV (\bZ^{(1)}_{S})^\top.
\end{align*}
Thus, we have
$ \Vert \mathbb{E} \bS_1 \bS_1^\top \Vert _2 \le  \Vert \mathbb{E} \bZ^{(1)}_{S} \bV (\bZ^{(1)}_{S})^\top \bZ^{(1)}_{S} \bV (\bZ^{(1)}_{S})^\top \Vert _2$.
Writing the $k$-th row of $\bZ^{(1)}$ as $Z^{(1)}_{k}$, the $(k,\ell)$-th entry of the above expectation is
$$
\mathbb{E} (Z^{(1)}_k)^\top \bV \sum_{i \in S} Z^{(1)}_i (Z^{(1)}_i)^\top \bV Z^{(1)}_\ell = \sum_{i \in S} \sum_{j_1, j_2, j_3, j_4} \bV_{j_1 j_2} \bV_{j_3 j_4} \mathbb{E} Z^{(1)}_{kj_1} Z^{(1)}_{ij_2} Z^{(1)}_{\ell j_3} Z^{(1)}_{ij_4}.
$$

When $k=\ell$, the only nonzero fourth moments in right-hand side arise from appropriate pairings of the $Z^{(1)}_{ab}$ terms.
\begin{align*}
   \mathbb{E} Z^{(1)}_{kj_1} Z^{(1)}_{ij_2} Z^{(1)}_{k j_3} Z^{(1)}_{ij_4} =
   \begin{cases}
       3 & i=k, j_1=j_2=j_3=j_4 \\
       1 & i=k, j_1=j_2 \neq j_3=j_4\\
       1 & i=k, j_1=j_3 \neq j_2=j_4\\
       1 & i=k, j_1=j_4 \neq j_2=j_3\\
       1 & i \neq k, j_1=j_3, j_2=j_4\\
       0 & otherwise.
   \end{cases}
\end{align*}
Hence,
$$
\mathbb{E} (Z^{(1)}_k)^\top \bV \Bigl(\sum_{i \in S} Z^{(1)}_i (Z^{(1)}_i)^\top \Bigr)\bV Z^{(1)}_k
= (|S|+1) \sum_{j_1, j_2} V_{j_1 j_2}^2 + (\sum_{j} V_{jj})^2 \leq(m+1)  \Vert \bV \Vert _F^2 + \text{tr}(\bV)^2.
$$
since $|S| \leq m$.

When $k \neq \ell$, all fourth moments $\mathbb{E} Z^{(1)}_{kj_1} Z^{(1)}_{ij_2} Z^{(1)}_{\ell j_3} Z^{(1)}_{ij_4}=0$. Therefore, $\Vert \mathbb{E} \bZ^{(1)}_{S} \bV (\bZ^{(1)}_{S})^\top \bZ^{(1)}_{S} \bV (\bZ^{(1)}_{S})^\top \Vert _2$ equals to the maximum diagonal entry, which is bounded above by $(m+1) \Vert \bV \Vert _F^2 + \text{tr}(\bV)^2$.

\subsubsection*{Bernstein inequality of $\bS_i$}
Assume that event $A$ holds. Then, for any $\delta>0$,
\begin{align*}
\mathbb{P}\left(	\bigg|\bigg|\sum_{i=1}^n \big\{\bZ^{(i)}_{S} \bV (\bZ^{(i)}_{S})^\top \text{I}_{A_i} - \mathbb{E}\bZ^{(i)}_{S} \bV (\bZ^{(i)}_{S})^\top \text{I}_{A_i} \big\}\bigg|\bigg|_2 \ge n \delta  \right) \qquad \qquad \qquad \qquad \qquad \qquad \\
\qquad\qquad\qquad\qquad\qquad \le 2m \exp\left\{
-\dfrac{1}{4} \min\bigg\{
\dfrac{n\delta^2}{(m+1)  \Vert \bV \Vert _F^2 + \text{tr}(\bV)^2},\dfrac{3n\delta}{2 \Vert \bV \Vert _2^2 K}
\bigg\} 
\right\}.
\end{align*}
Equivalently, for any $u>0$,
\begin{align*}
  \frac{1}{n}\bigg|\bigg|\sum_{i=1}^n \big\{\bZ^{(i)}_{S} \bV (\bZ^{(i)}_{S})^\top \text{I}_{A_i} - \mathbb{E}\bZ^{(i)}_{S} \bV (\bZ^{(i)}_{S})^\top \text{I}_{A_i} \big\}\bigg|\bigg|_2  \qquad \qquad \qquad \qquad \qquad \qquad \\
  \qquad\qquad\qquad\qquad\qquad \ge 4 \max\bigg\{\sqrt{\dfrac{\Big((m+1)  \Vert \bV \Vert _F^2 + \text{tr}(\bV)^2 \Big)u}{n}},\dfrac{2 \Vert \bV \Vert _2^2 K u}{3n}
\bigg\} 
\end{align*}
holds with probability at most $2m \exp(-u)$.

\subsubsection*{(2) Upper bound for $B_2$}
First, observe that
\begin{align*}
\bigg|\bigg|\mathbb{E}\bZ^{(1)}_{S} \bV (\bZ^{(1)}_{S})^\top \text{I}_{A_1^c}\bigg|\bigg|_2 & = 
\max\limits_{x:  \Vert x \Vert _2=1} \mathbb{E} (x^\top \bZ^{(1)}_{S} \bV (\bZ^{(1)}_{S})^\top x ) \cdot \text{I}_{A_1^c}\\
& \le 
\max\limits_{x:  \Vert x \Vert _2=1} \sqrt{\mathbb{E} (x^\top \bZ^{(1)}_{S} \bV (\bZ^{(1)}_{S})^\top x )^2 \cdot \mathbb{E} \text{I}_{A_1^c}}
\end{align*}
holds by the Cauchy-Schwartz inequality.

\subsubsection*{Calculation of $\mathbb{E} (x^\top \bZ^{(1)}_{S} \bV (\bZ^{(1)}_{S})^\top x )^2$}
For any unit vector $x \in \mathbb{R}^s$, note that
\begin{align*}
\mathbb{E} (x^\top \bZ^{(1)}_{S} \bV (\bZ^{(1)}_{S})^\top x )^2 & = \ \mathbb{E}\Bigl(
\sum\limits_{i_1, i_2 \in S}
\sum\limits_{j_1, j_2 \in [q]} x_{i_1} Z_{i_1 j_1}^{(1)} V_{j_1 j_2}Z_{i_2 j_2}^{(1)} x_{i_2}
\Bigr)^2
\\
& =
\sum_{\substack{i_1,i_2,i_3,i_4\in S\\j_1,\dots,j_4\in[q]}} x_{i_1} x_{i_2} x_{i_3} x_{i_4} V_{j_1 j_2} V_{j_3 j_4} \mathbb{E} Z_{i_1 j_1}^{(1)} Z_{i_2 j_2}^{(1)} Z_{i_3 j_3}^{(1)} Z_{i_4 j_4}^{(1)}.    
\end{align*}
The only nonzero fourth moments occur in the following four cases:
\begin{enumerate}
	\item $(i_1, j_1) = (i_2, j_2)$, $(i_3, j_3) = (i_4, j_4)$, and
	$(i_1, j_1) \neq (i_3, j_3)$
	$$
	\mathbb{E} Z_{i_1 j_1}^{(1)} Z_{i_2 j_2}^{(1)} Z_{i_3 j_3}^{(1)} Z_{i_4 j_4}^{(1)} = \mathbb{E}(Z_{i_1 j_1}^{(1)})^2 \mathbb{E}(Z_{i_3 j_3}^{(1)})^2 = 1
	$$
	\item $(i_1, j_1) = (i_3, j_3)$, $(i_2, j_2) = (i_4, j_4)$, and
	$(i_1, j_1) \neq (i_2, j_2)$ 
	$$
	\mathbb{E} Z_{i_1 j_1}^{(1)} Z_{i_2 j_2}^{(1)} Z_{i_3 j_3}^{(1)} Z_{i_4 j_4}^{(1)} = \mathbb{E}(Z_{i_1 j_1}^{(1)})^2 \mathbb{E}(Z_{i_2 j_2}^{(1)})^2 = 1
	$$
	\item $(i_1, j_1) = (i_4, j_4)$, $(i_3, j_3) = (i_2, j_2)$, and
	$(i_1, j_1) \neq (i_3, j_3)$ 
	$$
	\mathbb{E} Z_{i_1 j_1}^{(1)} Z_{i_2 j_2}^{(1)} Z_{i_3 j_3}^{(1)} Z_{i_4 j_4}^{(1)} = \mathbb{E}(Z_{i_1 j_1}^{(1)})^2 \mathbb{E}(Z_{i_3 j_3}^{(1)})^2 = 1
	$$	
	\item $(i_1, j_1) = (i_2, j_2)=(i_3, j_3) = (i_4, j_4)$
	$$
	\mathbb{E} Z_{i_1 j_1}^{(1)} Z_{i_2 j_2}^{(1)} Z_{i_3 j_3}^{(1)} Z_{i_4 j_4}^{(1)} = \mathbb{E}(Z_{i_1 j_1}^{(1)})^4 = 3
	$$
\end{enumerate}
Therefore, we have
\begin{align*}
    \mathbb{E} (x^\top \bZ^{(1)}_{S} \bV (\bZ^{(1)}_{S})^\top x )^2
    =&\ \sum\limits_{i_1, i_2 \in S}
    \sum\limits_{j_1, j_2 \in [q]} x_{i_1}^2 x_{i_2}^2(V_{j_1 j_1} V_{j_2 j_2} + 2V_{j_1 j_2}^2)
    + 3\sum\limits_{i_1 \in S}
    \sum\limits_{j_1 \in [q]} x_{i_1}^4 V_{j_1 j_1}^2\\
    \leq&\ \sum\limits_{j_1, j_2 \in [q]} (V_{j_1 j_1} V_{j_2 j_2} + 2V_{j_1 j_2}^2)
    + 3\sum\limits_{j_1 \in [q]} V_{j_1 j_1}^2
    =\ \text{tr}(\bV)^2 + 2 \Vert \bV \Vert _F^2.
\end{align*}

\subsubsection*{Calculation of $\text{\rm P}(A_1^c)$}

Note that $ \Vert \bZ_S^{(1)} \Vert _2= \Vert (\bZ_S^{(1)})^\top \Vert _2 = \sigma_{\max}((\bZ_S^{(1)})^\top)$. 
By Theorem 6.1 of \citet{wainwright2019high}, 
$$
\mathbb{P}\bigl(
\sigma_{\max}((\bZ_S^{(1)})^\top) \ge \sqrt{q} (1 + \delta) + \sqrt{m} \bigr) 
\le \exp(-q\delta^2 / 2), \quad \forall\,\delta>0.
$$
Here $S$ is treated as fixed, so the $q$–dimension plays the role of sample size.
Setting $\delta = \sqrt{2(u + \log n) /q}$ for $u>0$ gives
$$
\mathbb{P}\bigl(
\sigma_{\max}((\bZ_S^{(1)})^\top) \ge \sqrt{q} + \sqrt{m} + \sqrt{2(u + \log n)}
\bigr) \le \frac{1}{n}\exp(-u), \quad u>0.
$$
By a union bound over $i=1,\dots,n$, it follows that
$$
\mathbb{P}\left(\cup_{i=1}^n \Big\{ \Vert \bZ_S^{(i)} \Vert _2^2 \ge K_u \Big\} \right) \le 
n\, \mathbb{P}\left( \Vert \bZ_S^{(1)} \Vert _2^2 \ge K_u \right) \le \exp(-u),
$$
where $K_u = \sqrt{q} + \sqrt{m} + \sqrt{2(u + \log n)}$.	
Hence, by choosing $u \ge \log n$, we get
$$
\bigg|\bigg|\mathbb{E}\bZ^{(1)}_{S} \bV (\bZ^{(1)}_{S})^\top \text{I}_{A_1^c}\bigg|\bigg|_2
\le
\sqrt{
	\text{tr}(\bV)^2 + 2 \Vert \bV \Vert _F^2}\ e^{-u / 2} \le 
\sqrt{\dfrac{\text{tr}(\bV)^2 + 2 \Vert \bV \Vert _F^2}{n}}.
$$

\subsubsection*{(3) Combined results}
Now, define
$$
t_1 = 4 \max\bigg\{
\sqrt{\dfrac{\Big((m+1)  \Vert \bV \Vert _F^2 + \text{tr}(\bV)^2 \Big)u}{n}},\dfrac{2 \Vert \bV \Vert _2^2 K_u \cdot u}{3n}
\bigg\}, \quad t_2 = \sqrt{\dfrac{\text{tr}(\bV)^2 + 2 \Vert \bV \Vert _F^2}{n}},
$$
where $K_u=\sqrt{q} + \sqrt{m} + \sqrt{2(u+\log n)}$
and set $A \equiv \cap_{i=1}^n \{ \Vert \bZ_S^{(i)} \Vert _2^2 \le K_u\}$. Then
\begin{align*}
&\mathbb{P}\left[
\bigg|\bigg|\frac{1}{n}\sum_{i=1}^n \big\{\bZ^{(i)}_{S} \bV (\bZ^{(i)}_{S})^\top - \mathbb{E}\bZ^{(i)}_{S} \bV (\bZ^{(i)}_{S})^\top \big\}\bigg|\bigg|_2 \ge t_1 + t_2
\right] \\
\le\, &\mathbb{P}\left[
\bigg|\bigg|\frac{1}{n}\sum_{i=1}^n \big\{\bZ^{(i)}_{S} \bV (\bZ^{(i)}_{S})^\top - \mathbb{E}\bZ^{(i)}_{S} \bV (\bZ^{(i)}_{S})^\top \big\}\bigg|\bigg|_2 \ge t_1 + t_2 \bigg| A \right] + \mathbb{P}(A^c)\\
\le\, &\mathbb{P}\left[
\frac{1}{n}\bigg|\bigg|\sum_{i=1}^n \big\{\bZ^{(i)}_{S} \bV (\bZ^{(i)}_{S})^\top \text{I}_{A_i} - \mathbb{E}\bZ^{(i)}_{S} \bV (\bZ^{(i)}_{S})^\top \text{I}_{A_i} \big\}\bigg|\bigg|_2\ge t_1 \bigg| A \right] + \exp(-u)\\
\le\, &3m \exp(-u)    
\end{align*}
Therefore, for $u > 1 \vee \log n$, we have
\begin{equation}
  \begin{aligned} \label{Gamma concentration inequality}
    \Vert \hat{\bGamma}_{SS} - {\bGamma}_{SS} \Vert _2
    \geq&\, C \cdot  \Vert \bU \Vert _2 \cdot \max\bigg\{
\sqrt{\dfrac{\Big((m+1)  \Vert \bV \Vert _F^2 + \text{tr}(\bV)^2\Big)u }{nq^2}},
\dfrac{2 \Vert \bV \Vert _2^2 K_u \cdot u}{3nq} \bigg\} \\
\geq&\, C \cdot  \Vert \bU \Vert _2 \cdot \max\bigg\{
\sqrt{\dfrac{\Big((m+1)  \Vert \bV \Vert _F^2 + \text{tr}(\bV)^2\Big)u }{nq^2}},
\dfrac{4 \Vert \bV \Vert _2^2 (u^{\frac{3}{2}}+u\sqrt{q})}{3nq} \bigg\}
 \end{aligned}  
\end{equation}
with probability at most $3m \exp(-u)$.
Since $K_u \leq 2\sqrt{q}+2\sqrt{u}$, it follows that
$$
\Vert \hat{\bGamma}_{SS} - {\bGamma}_{SS} \Vert _2 
\geq 2 C \cdot  \Vert \bU \Vert _2 \cdot 
\max\big\{
a_1 \sqrt{u},\
a_2 u,\
a_3 u\sqrt{u}
\big\}
$$
holds with the same probability, where 
\begin{align*}
    a_1 = \sqrt{(m+1) \Vert \bV \Vert _F^2 + \text{tr}(\bV)^2}/(\sqrt{n}q),\quad
    a_2 = \Vert \bV \Vert_2^2 / (n\sqrt{q}), \quad
    a_3 = \Vert \bV \Vert_2^2 / (nq).	    
\end{align*}
This leads to 
$$
\mathbb{P} \left[
\Vert \hat{\bGamma}_{SS} - {\bGamma}_{SS} \Vert _2 
\geq 2 C \cdot  \Vert \bU \Vert _2 \cdot t\right]
\le 
3m \exp \bigg[
- \min \Big\{
(t / a_1)^2, t / a_2, (t / a_3)^{2/3}
\Big\}
\bigg].
$$
Finally, using $\|\bV\|_F^2\le q\|\bV\|_2^2$ and $\text{tr}(\bV)/q\le\|\bV\|_2$ shows $a_1 \le \sqrt{(2m/q+1) \,\|\bV\|_2^2/n}$, which completes the proof.
\end{proof}

\subsection{Lemma \ref{lem:invcov_dev}}
\begin{lemma}\label{lem:invcov_dev}
    Let $\bSigma$ and $\hat{\bSigma}$ be nonsingular matrices, and suppose $\|\bSigma^{-1}\| \cdot \|\widehat{\bSigma} - \bSigma\| < 1$. Then for any sub-multiplicative matrix norm,
\begin{equation}\nonumber
    \|\widehat{\bSigma}^{-1} - \bSigma^{-1} \| \le \frac{\|\bSigma^{-1}\|^2 \cdot \|\widehat{\bSigma} - \bSigma\|}{1-\|\bSigma^{-1}\| \cdot \|\widehat{\bSigma} - \bSigma\| }.
\end{equation}
\end{lemma}

\begin{proof}
Observe that
\begin{align*}
	\|\widehat{\bSigma}^{-1} - \bSigma^{-1} \|
	&= \|\bSigma^{-1}(\widehat{\bSigma} - \bSigma)\widehat{\bSigma}^{-1}\| \\
	&\le \|\bSigma^{-1}\| \cdot \|\widehat{\bSigma} - \bSigma\| \cdot \|\widehat{\bSigma}^{-1}\| \\
	&\le \|\bSigma^{-1}\| \cdot \|\widehat{\bSigma} - \bSigma\| \cdot \big(\|\widehat{\bSigma}^{-1} - \bSigma^{-1}\| + \|\bSigma^{-1}\| \big)
\end{align*}
By rearranging and using the assumption $\|\bSigma^{-1}\| \cdot \|\widehat{\bSigma} - \bSigma\| < 1$, we yield the result.
\end{proof}

\section{Additional results for simulation study} \label{appendix:numerical}

\subsection{Results - homongenous variances} \label{appendix:homogeneous}
We first present simulation results under homogeneous variances, which correspond to the setting of the main experiment in Section 4, except that the diagonal entries $\omega_{ii}$ are not scaled by uniform random numbers. All other settings follow analogously to Section 4.

\begin{figure}[htbp]
  \centering
  \includegraphics[width=\textwidth]{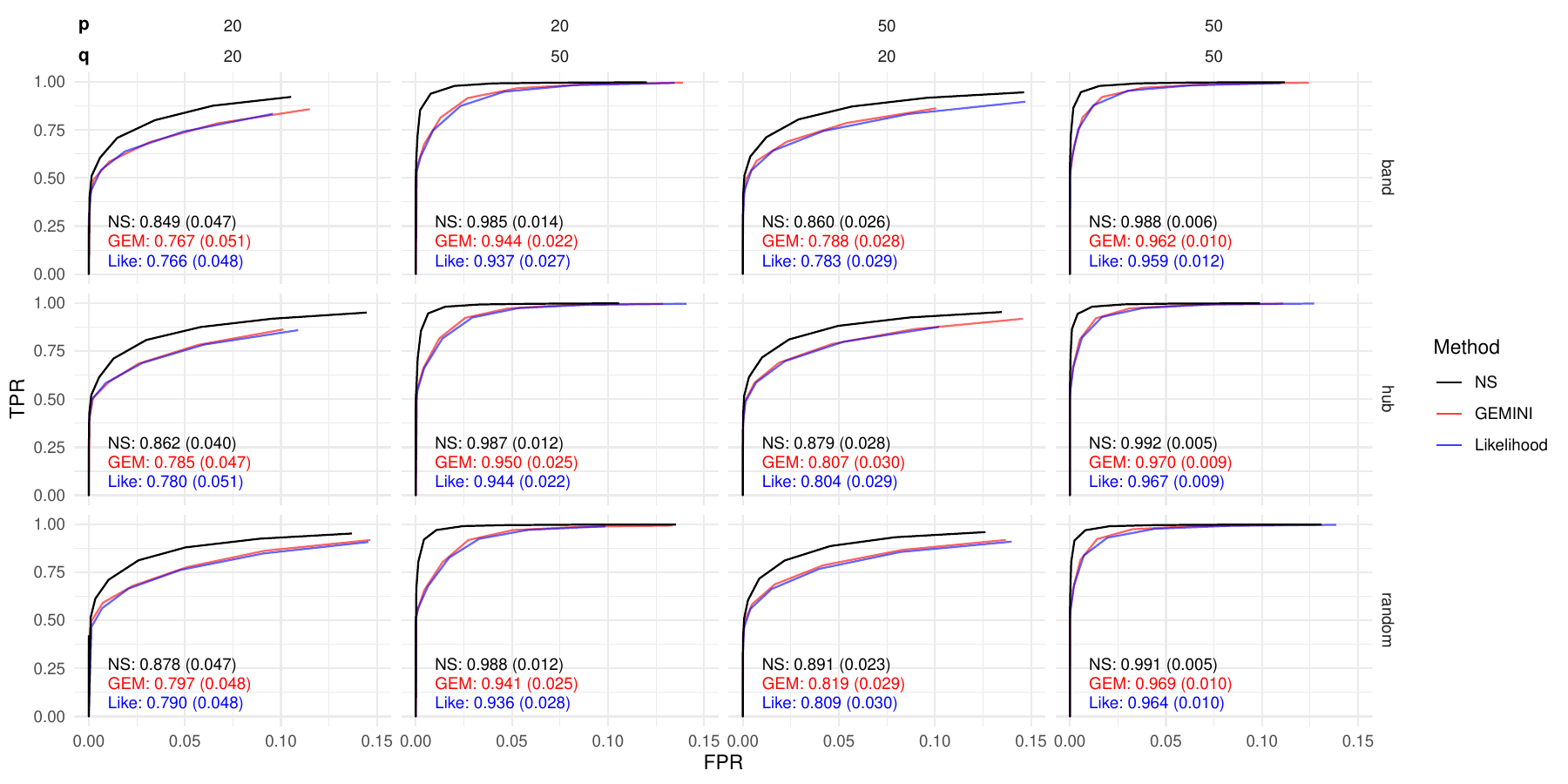}
   \includegraphics[width=\textwidth]{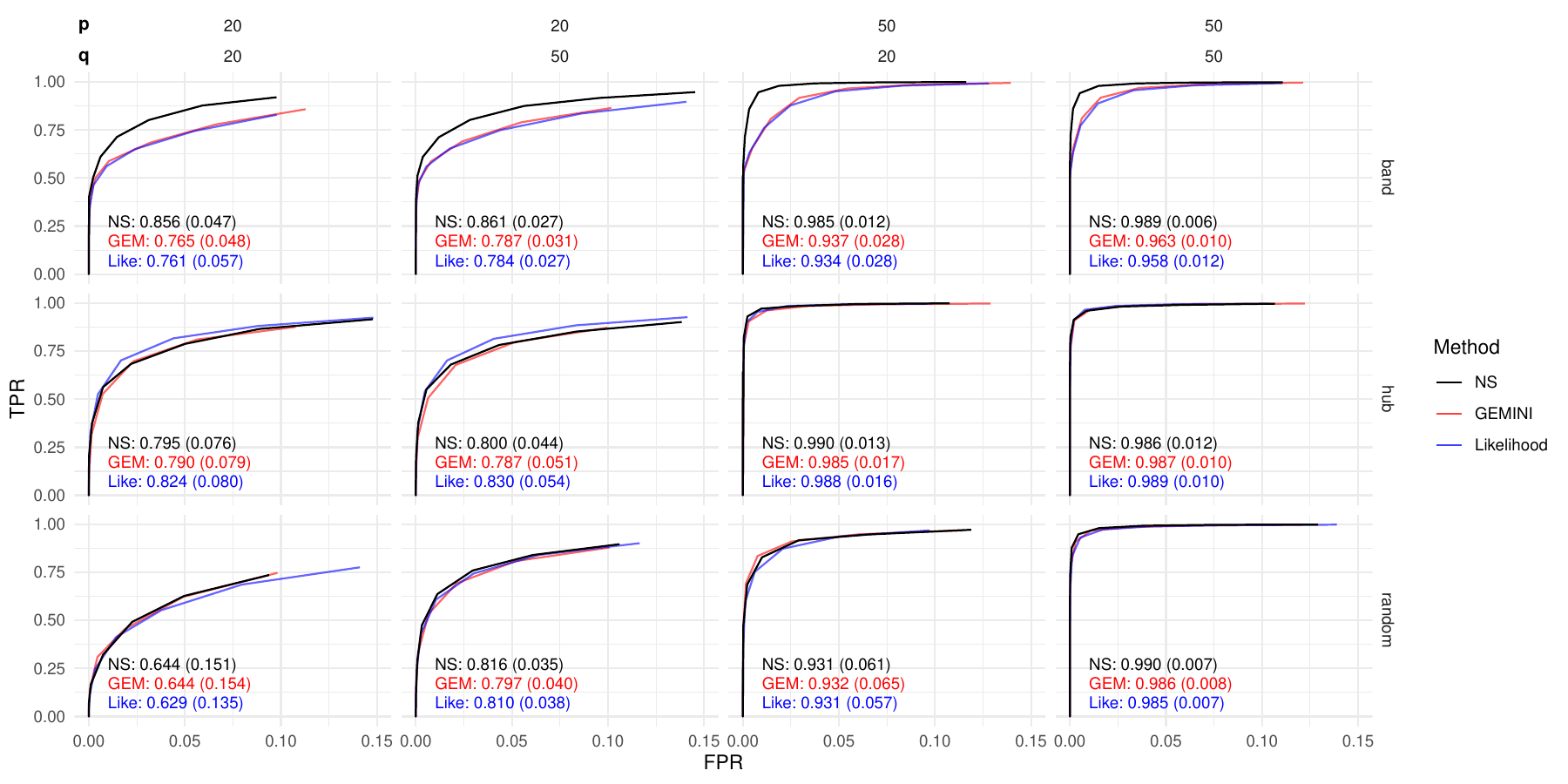}
  \caption{The ROC curves and partial AUC values for row-wise (top) and column-wise (bottom) graph estimation under homogeneous variances, comparing various methods. Horizontal labels denote $(p,q)$ combinations; vertical labels indicate the true graph structure. Text within each panel gives the partial AUC and its standard deviation (in parentheses).}
  \label{fig:sim_homo}
\end{figure}

Figure \ref{fig:sim_homo} demonstrates the estimation performance under homogeneous variance. Similar to the main results in Section 4.2, our method shows a clear advantage when estimating the banded structure, while providing comparable performance to other methods for hub and random structures.

\subsection{Results - sample size} \label{appendix:samplesize}
Next, we investigate the estimation performance as the sample size varies. To this end, we fix $p$ and $q$ at $20$, and compare the results for $n=20$, $50$, and $200$. Other settings, including graph structures and signal strength $\rho$, remain identical to those in the main experiment.

The results are presented in Figure \ref{fig:sim_samplesize}. As the sample size $n$ increases, overall estimation performance improves, and our matrixNS remains the best performing method across all settings. Combined with the findings from the main experiment, these results support that the effective sample size—$nq$ for row-wise estimation and $np$ for column-wise estimation—is a key factor of graph recovery performance.

\begin{figure}[htbp]
  \centering
  \includegraphics[width=\textwidth]{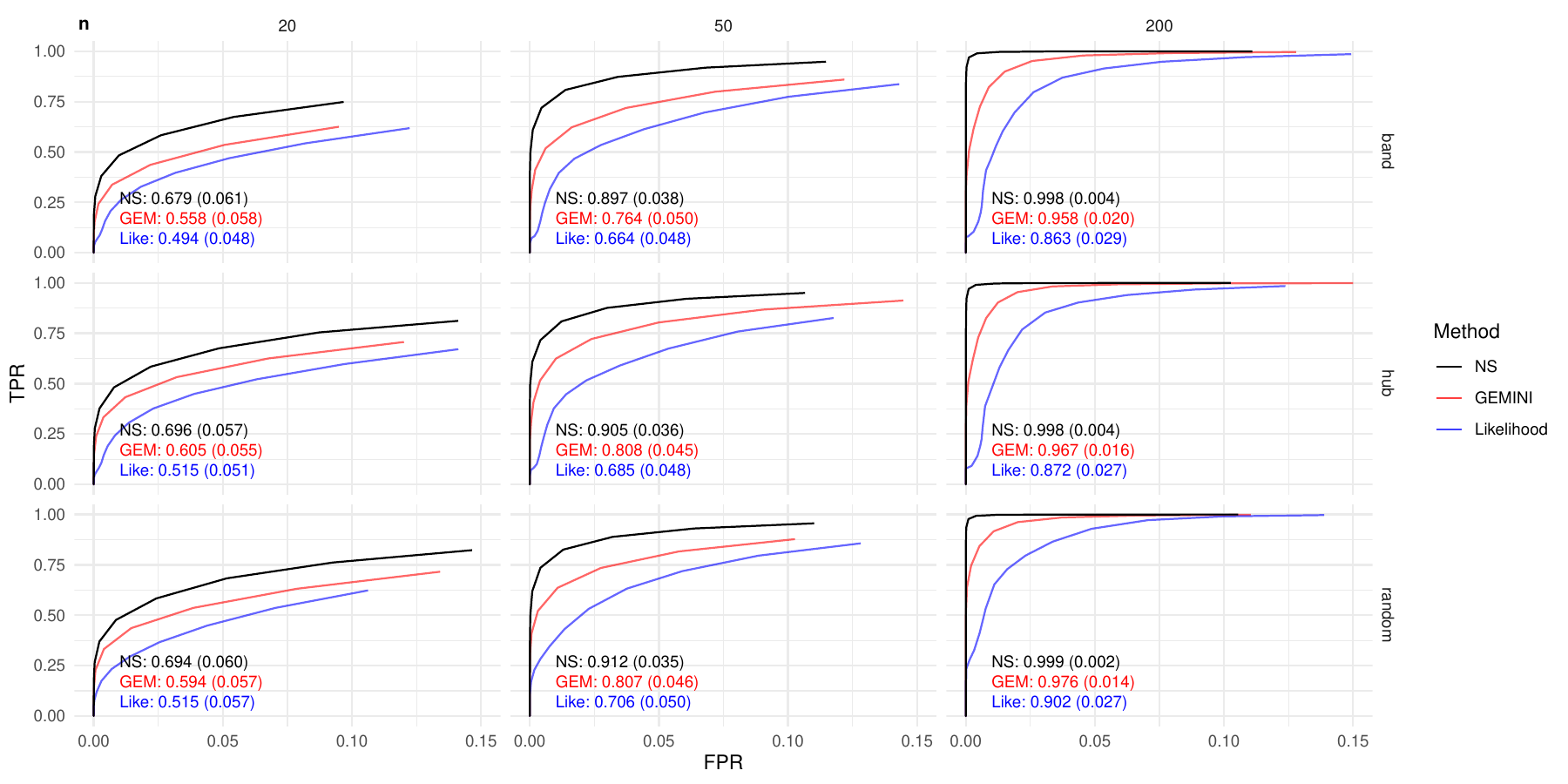}
   \includegraphics[width=\textwidth]{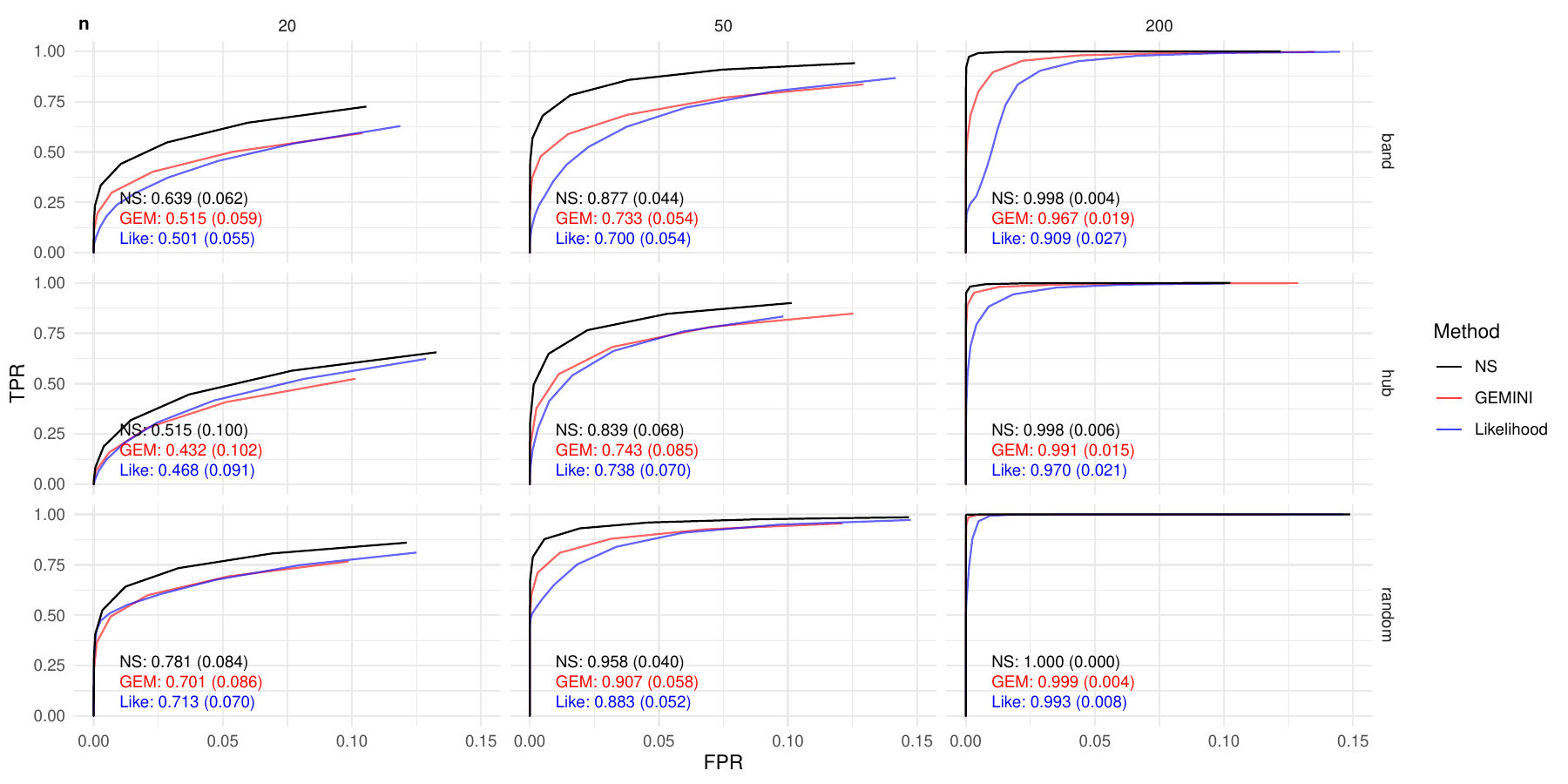}
  \caption{The ROC curves and partial AUC values for row-wise (top) and column-wise (bottom) graph estimation, under various sample size. Data dimension $(p,q)$ are fixed to $(20,20)$.}
  \label{fig:sim_samplesize}
\end{figure}

\subsection{Results - signal strength} \label{appendix:signalstrength}
We then examine the effect of signal strength by adjusting the signal size parameter. Specifically, when constructing $\bU^{-1}, \bV^{-1}$, we set three levels (``low'', ``moderate'', ``high'') of $\rho$ as
$$
\rho_{\text{hub}} = (0.3, 0.4, 0.6),\quad
\rho_{\text{band}} = (0.4, 0.6, 0.8),\quad
\rho_{\text{rand}} = (0.2, 0.4, 0.6).
$$
We set $n=20$ and $(p,q)=(50,50)$, with all other experimental conditions remain unchanged. As before, we fix the row covariance matrix to the band structure with $\rho_{\text{band}}=0.4$, while varying the column structure and its signal size. Since the row covariance structure is fixed, the row graph selection results are identical to those in the main experiment; thus, we only report the results from column graph selection.

\begin{figure}[htbp]
  \centering
  \includegraphics[width=\textwidth]{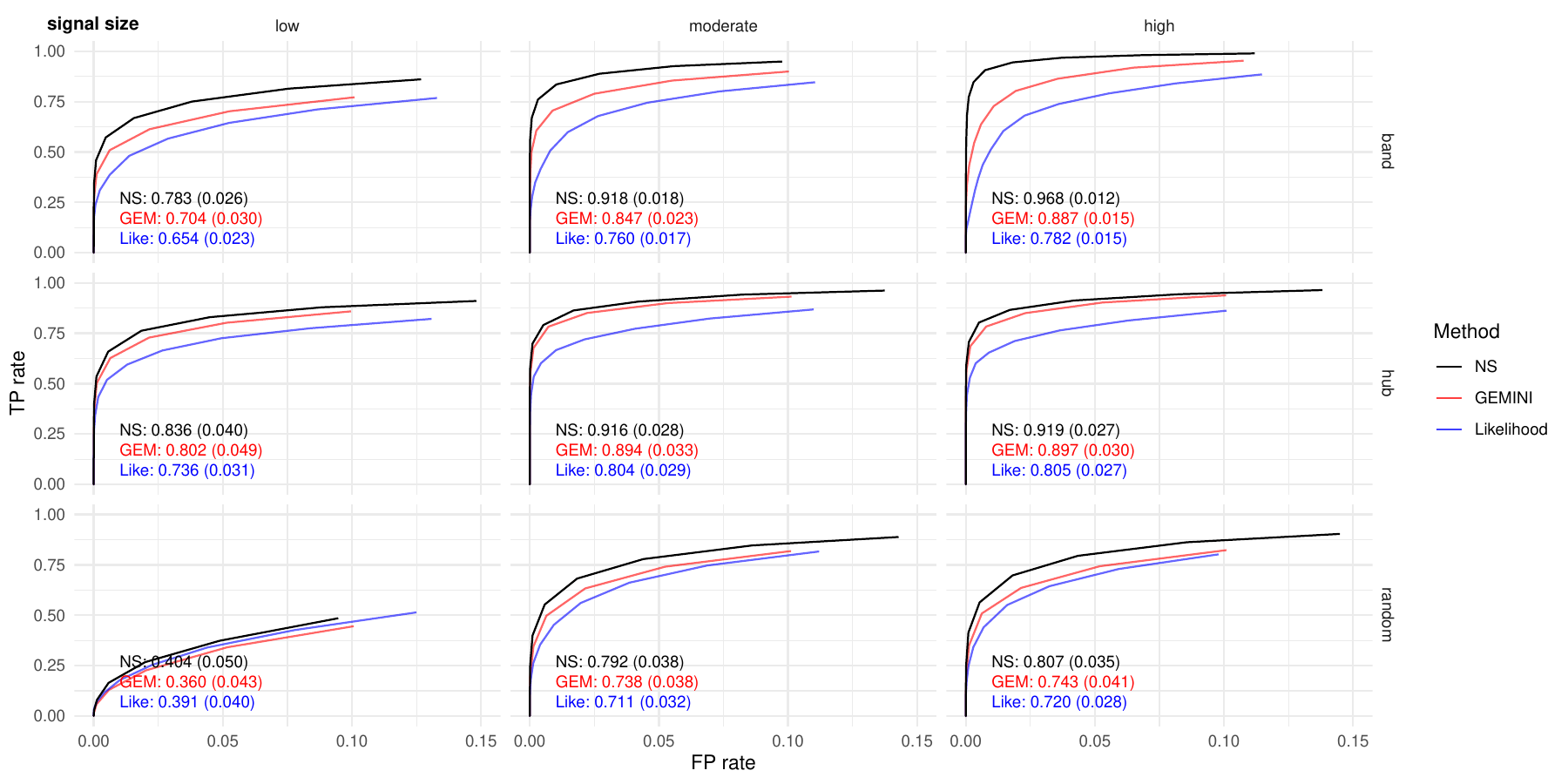}
  \caption{The ROC curves and partial AUC values for the column-wise graph estimation, under various signal strength. The horizontal axis denotes the signal strength parameter, increasing from left (``low'') to right (``high'').}
  \label{fig:sim_signalsize_col}
\end{figure}

Figure \ref{fig:sim_signalsize_col} shows that the overall estimation accuracy increases as the signal strength increases, and our proposed method always outperforms the other methods.

\end{document}